\documentclass{jgaa-art}
\usepackage[usenames]{xcolor} 

\usepackage{graphicx, wrapfig, subfig}
\graphicspath{{figures/}}
\usepackage{enumerate, hyperref, cite}

\usepackage{amsmath, amssymb}
\allowdisplaybreaks[1]
\newtheorem{corollary}{Corollary}
\newtheorem{definition}[lemma]{Definition}

\DeclareMathOperator{\chessboard}{chessboard} 
\DeclareMathOperator{\parent}{parent} 
\DeclareMathOperator{\ctree}{ctree} 
\DeclareMathOperator{\pre}{pre} 
\DeclareMathOperator{\rpost}{post} 
\DeclareMathOperator{\cperm}{cperm} 
\DeclareMathOperator{\stretchperm}{stretch} 
\DeclareMathOperator{\augment}{augment} 

\newcommand{\R}{\mathbf{R}}

\begin{document}
\doi{} 
\Issue{0}{0}{0}{0}{0} 
\HeadingAuthor{Bannister et al.} 
\HeadingTitle{Superpatterns and Universal Point Sets} 
\title{Superpatterns and Universal Point Sets}
\Ack{This research was supported in part by the National Science Foundation under grants 0830403 and 1217322, and by the Office of Naval Research under MURI grant N00014-08-1-1015.}
\author{Michael J. Bannister}{mbannist@uci.edu}
\author{Zhanpeng Cheng}{zhanpenc@uci.edu}
\author{William E. Devanny}{wdevanny@uci.edu}
\author{David Eppstein}{eppstein@uci.edu}
\affiliation{Department of Computer Science, University of California, Irvine}

\submitted{December 2013}%
\reviewed{February 2014}%
\revised{March 2014}%
\accepted{March 2014}%
\final{April 2014}%
\published{}%
\type{Regular Paper}%
\editor{Stephen Wismath, Alexander Wolff}%

\maketitle

\begin{abstract}
An old open problem in graph drawing asks for the size of a \emph{universal point set}, a set of points that can be used as vertices for straight-line drawings of all $n$-vertex planar graphs.
We connect this problem to the theory of permutation patterns, where another open problem concerns the size of \emph{superpatterns}, permutations that contain all patterns of a given size. We generalize superpatterns to classes of permutations determined by forbidden patterns, and we construct superpatterns of size $n^2/4+\Theta(n)$ for the $213$-avoiding permutations, half the size of known superpatterns for unconstrained permutations. We use our superpatterns to construct universal point sets of size $n^2/4-\Theta(n)$, smaller than the previous bound by a 9/16 factor. We prove that every proper subclass of the $213$-avoiding permutations has superpatterns of size $O(n\log^{O(1)}n)$, which we use to prove that the planar graphs of bounded pathwidth have near-linear universal point sets.
\end{abstract}

\Body

\section{Introduction}
Fary's theorem tells us that every planar graph can be drawn with its edges as non-crossing straight line segments. As usually stated, this theorem allows the vertex coordinates of the drawing to be drawn from an uncountable and unbounded set, the set of all points in the plane. It is natural to ask how tightly we can constrain the set of possible vertices. In this direction, the \emph{universal point set problem} asks for a sequence of point sets $U_n \subseteq \R^2$ such that every planar graph with $n$ vertices can be straight-line embedded with vertices in $U_n$ and such that the size of $U_n$ is as small as possible.

So far the best known upper bounds for this problem have considered sets $U_n$ of a special form: the intersection of the integer lattice with the interior of a convex polygon. In 1988 de Fraysseix, Pach and Pollack showed that a triangular set of lattice points within a rectangular grid of $(2n-3)\times (n-1)$ points forms a universal set of size $n^2-O(n)$~\cite{FraPacPol-IJC-1990, ChrPay-IPL-1995}, and in 1990 Schnyder found more compact grid drawings within the lower left triangle of an $(n-1) \times (n-1)$ grid\cite{Sch-SODA-1990}, a set of size $n^2/2 - O(n)$. Using the method of de Fraysseix et al., Brandenburg found that a triangular subset of a $\frac{4}{3}n \times \frac{2}{3}n$ grid, of size $\frac{4}{9}n^2 + O(n)$, is universal~\cite{Bra-ENDM-08}. Until now his bound has remained the best known.

On the other side, Dolev, Leighton, and Trickey~\cite{DolLeiTri-AiCR-84} used the \emph{nested triangles graph} to show that rectangular grids that are universal must have size at least  $n/3 \times n/3$, and that grids that are universal for drawings with a fixed choice of planar embedding and outer face must have size at least $2n/3\times 2n/3$. Thus, if we wish to find subquadratic universal point sets we must consider sets not forming a grid. However, the known lower bounds that do not make this grid assumption are considerably weaker. In 1988 de Fraysseix, Pach and Pollack proved the first nontrivial lower bounds of $n + \Omega(\sqrt{n})$ for a general universal point set~\cite{FraPacPol-IJC-1990}. This was later improved to $1.098n-o(n)$ by Chrobak and Payne~\cite{ChrPay-IPL-1995}. Finally, Kurowski improved the lower bound to $1.235n$~\cite{Kur-IPL-04}, which is still the best  lower bound known.\footnote{The validity of this result was questioned by Mondal~\cite{Mon-MS-12}, but later confirmed.}

With such a large gap between these lower bounds and Brandenburg's upper bound, obtaining tighter bounds remains an important open problem in graph drawing~\cite{TOPP-45}.

Universal point sets have also been considered for subclasses of planar graphs. For instance, every set of $n$ points in general position (no three collinear) is universal for the $n$-vertex outerplanar graphs~\cite{GriMohPol-AMM-1991}. Simply-nested planar graphs (graphs that can be decomposed into nested induced cycles) have universal point sets of size $O\!\left(n (\log n / \log\log n)^2\right)$~\cite{AngDibKau-GD-2012}, and planar 3-trees have universal point sets of size $O(n^{3/2}\log n)$ \cite{FulTot-WADS-2013}. Based in part on the results in this paper, the graphs of simple line and pseudoline arrangements have been shown to have universal point sets of size $O(n\log n)$~\cite{Epp-GD-13}.

In this paper we provide a new upper bound on universal point sets for general planar graphs, and improved bounds for certain restricted classes of planar graphs. We approach these problems via a novel connection to a different field of study than graph drawing, the study of patterns in permutations.\footnote{A different connection between permutation patterns and graph drawing is being pursued independently by Bereg, Holroyd, Nachmanson, and Pupyrev, in connection with bend minimization in bundles of edges that realize specified permutations~\cite{BerHolNac-12}.} A permutation $\sigma$ is said to \emph{contain} the pattern $\pi$ (also a permutation) if $\sigma$ has a (not necessarily contiguous) subsequence with the same length as $\pi$, whose elements are in the same relative order with respect to each other as the corresponding elements of $\pi$. The permutations that contain none of the patterns in a given set $F$ of forbidden patterns are said to be \emph{$F$-avoiding}; we define $S_n(F)$ to be the set of length-$n$ permutations that avoid $F$. Researchers in permutation patterns have defined a \emph{superpattern} to be a permutation that contains all length-$n$ permutations among its patterns, and have studied bounds on the lengths of these patterns~\cite{Arr-EJC-99,ErkErkSva-AC-2007}, culminating in a proof by Miller that there exist superpatterns of length $n^2 / 2 + \Theta(n)$ \cite{Mil-JCTSA-2009}. We generalize this concept to an $S_n(F)$-superpattern, a permutation that contains all possible patterns in $S_n(F)$; we prove that for certain sets $F$, the $S_n(F)$-superpatterns are much shorter than Miller's bound.

As we show, the existence of small $S_n(213)$-superpatterns leads directly to small universal point sets for arbitrary planar graphs. In the same way, the existence of small $S_n(F)$-super\-patterns for forbidden pattern sets $F$ that contain $213$ leads to small universal point sets for subclasses of the planar graphs. Our method constructs a universal set $U$ that has one point for each element of the superpattern $\sigma$. It uses two different traversals of a depth-first-search tree of a canonically oriented planar graph $G$ to derive a permutation $\cperm(G)$ from $G$, and it uses the universality of $\sigma$ to find $\cperm(G)$ as a pattern in $\sigma$. Then, the positions of the elements of this pattern in $\sigma$ determine the assignment of the corresponding vertices of $G$ to points in $U$, and we prove that this assignment gives a planar drawing of $G$.

Specifically our contributions include proving the existence of:
\begin{itemize}
\item superpatterns for $213$-avoiding permutations of size $n^2/4 + \Theta(n)$;
\item universal point sets for planar graphs of size $n^2 / 4 - \Theta(n)$;
\item superpatterns for every proper subclass of the $213$-avoiding permutations of size $O(n\log^{O(1)}n)$;
\item universal point sets for graphs of bounded pathwidth of size $O(n\log^{O(1)}n)$; and
\item universal point sets for simply-nested planar graphs of size $O(n\log n)$.
\end{itemize}
In addition, we prove that every superpattern for $\{213, 132\}$-avoiding permutations has length $\Omega(n\log n)$, which in turn implies that every superpattern for $213$-avoiding permutations has length $\Omega(n\log n)$. It was known that $S_n$-superpatterns must have quadratic length---otherwise they would have too few length-$n$ subsequences to cover all $n!$ permutations~\cite{Arr-EJC-99}---but such counting arguments cannot provide nonlinear bounds for $S_n(F)$-superpatterns due to the now-proven Stanley--Wilf conjecture that $S_n(F)$ grows singly exponentially~\cite{MarTar-JCTA-04}. Instead, our proof finds an explicit set of $\{213, 132\}$-avoiding permutations whose copies within a superpattern cannot share many elements.

A similar but simpler reduction uses $S_n$-superpatterns to construct universal point sets for \emph{dominance drawings} of transitively reduced $st$-planar graphs.
In subsequent work \cite{BanDevEpp-ANALCO-14}, we study dominance drawings based on superpatterns for $321$-avoiding permutations and their subclasses, and we relate these subclasses to natural classes of $st$-planar graphs and nonplanar Hasse diagrams of width-2 partial orders.

\section{Preliminaries}

\subsection{Permutation patterns}
Let $S_n$ denote the set of all \emph{permutations} of the numbers from $1$ to $n$. We will normally specify a permutation as a sequence of numbers: for instance, the six permutations in $S_3$ are $123$, $132$, $213$, $231$, $312$, and $321$. If $\pi$ is a permutation, then we write $\pi_i$ for the element in the $i$th position of $\pi$,
and $|\pi|$ for the number of elements in $\pi$. The \emph{inverse} of a permutation $\pi$ is a permutation $\pi^{-1}$ such that $\pi_i = j$ if and only if $(\pi^{-1})_j = i$ for all $1 \leq i,j \leq |\pi|$.

\begin{definition}
A permutation $\pi$ is a \emph{pattern} of a permutation $\sigma$ of length $n$ if there exists a sequence of integers $1\le \ell_1 < \ell_2 < \cdots < \ell_{|\pi|}\le n$ such that $\pi_i<\pi_j$ if and only if $\sigma_{\ell_i} < \sigma_{\ell_j}$ for every $1 \le i, j \le |\pi|$.
In other words, $\pi$ is a pattern of $\sigma$ if $\pi$ is order-isomorphic to a subsequence of $\sigma$. We say that a permutation $\sigma$ \emph{avoids} a permutation $\phi$ if $\sigma$ does not contain $\phi$ as a pattern.
\end{definition}

\begin{definition}
A \emph{permutation class} is a set of permutations with the property that all patterns of all permutations in the class also belong to the class.
\end{definition}

Every permutation class may be defined by a set of \emph{forbidden patterns}, the minimal permutations that do not belong to the class; however, this set might not be finite.

\begin{definition}
$S_n(\phi_1, \ldots, \phi_k)$ denotes the set of all length-$n$ permutations that avoid all of the (forbidden) patterns $\phi_1, \ldots, \phi_k$.
\end{definition}

\begin{definition}
A \emph{$P$-superpattern}, for a set of permutations $P \subseteq S_n$, is a permutation $\sigma$ with the property that every $\pi \in P$ is a pattern of $\sigma$.
\end{definition}

One of the most important permutation classes in the study of permutation patterns is the class of \emph{stack-sortable permutations}~\cite{Rot-DM-81}, the permutations that avoid the pattern $231$. Knuth's discovery that these are exactly the permutations that can be sorted using a single stack~\cite{Knu-TAoCP-68} kicked off the study of permutation patterns. The $213$-avoiding permutations that form the focus of our research are related to the $231$-avoiding permutations by a simple transformation, the replacement of each value $i$ in a permutation by the value $n+1-i$, that does not affect the existence or size of superpatterns.

\subsection{Chessboard representation}

\begin{definition}
A \emph{run} of a permutation is a contiguous monotone subsequence of the permutation. A \emph{block} of a permutation is a set of consecutive integers that appear contiguously (but not necessarily in order) in~$\pi$. For instance, $135$ is a run in $21354$ and $\{3,4,5\}$ is a block in $14352$.
\end{definition}

\begin{definition}
Given a permutation $\pi$, a \emph{column} of $\pi$ is a maximal ascending run of~$\pi$ and a \emph{row} of~$\pi$ is a maximal ascending run in $\pi^{-1}$.
\end{definition}

A slightly different definition of rows and columns was used by Miller~\cite{Mil-JCTSA-2009}. For our definition, the intersection of a row and column is a block that could contain more than one element, whereas in Miller's definition a row and column necessarily intersect in at most one element.

\begin{definition}
The \emph{chessboard representation} of a permutation $\pi$ is an $r \times c$ matrix
\[M=\chessboard(\pi),\]
where $r$ is number of rows in $\pi$ and $c$ is the number of columns in $\pi$, such that $M(i,j)$ is the number of points in the intersection of the $i^\text{th}$ column and the $j^\text{th}$ row of $\pi$.
\end{definition}

An example of a chessboard representation can be seen in Figure~\ref{fig:chess-examp2}.  To recover a permutation from its chessboard representation, start with the lowest row and work upwards assigning an ascending subsequence of values to the squares of each row in left to right order within each row.  If a square has label $i$, allocate $i$ values for it.  Then, after this assignment has been made, traverse each column in left-to-right order, within each column listing in ascending order the values assigned to each square of the column. The sequence of values listed by this column traversal is the desired permutation.
\begin{figure}[ht]
\centering
\includegraphics[width=0.9\textwidth]{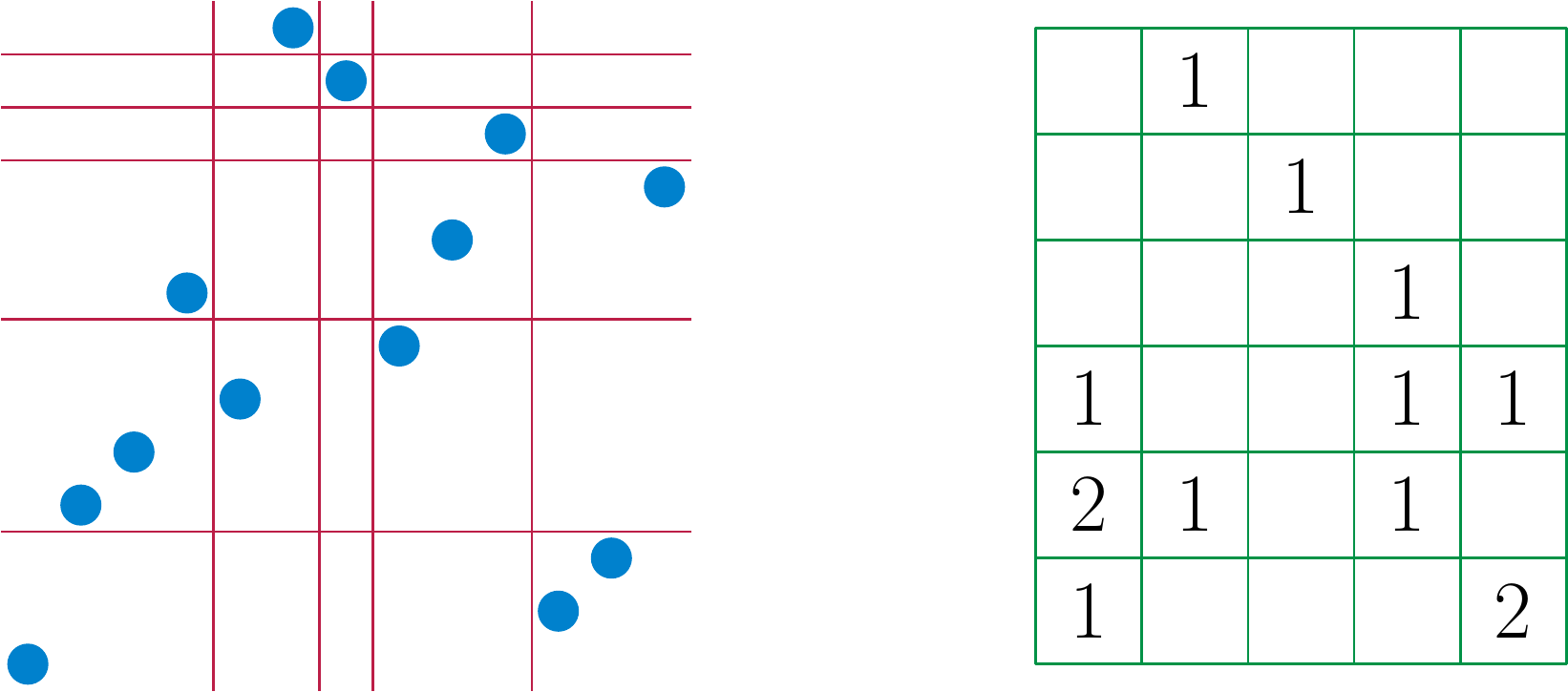}
\caption{The permutation $\pi = 1\ 4\ 5\ 8\ 6\ 13\ 12\ 7\ 9\ 11\ 2\ 3\ 10$ represented by its scatterplot (the points $(i,\pi_i)$) with lines separating its rows and columns (left), and by $\chessboard(\pi)$ (right).}
\label{fig:chess-examp2}
\end{figure}

\subsection{Subsequence majorization}
\label{sec:subsequence-majorization}

Several of our results use the fact that every sequence of positive integers can be majorized (termwise dominated) by a subsequence of another sequence with slowly growing sums,
as we detail in this section.

\begin{definition}
Let $\xi$ be the sequence of numbers $\xi_i$ ($i=1,2,3,\dots)$ given by the expression $\xi_i=i\oplus(i-1)$, where $\oplus$ denotes bitwise binary exclusive or.
\end{definition}

This gives the sequence
$$1,3,1,7,1,3,1,15,1,3,1,7,1,3,1,31,\dots$$
in which each value is one less than a power of two. This sequence has a recursive doubling construction: for any power of two, say $q$, the first $2q-1$ values of $\xi$
consist of two repetitions of the first $q-1$ values, separated by the value $2q-1$.
We confirm a conjecture of Klaus Brockhaus from 2003\footnote{See sequence A080277 in the Online Encyclopedia of Integer Sequences, \url{http://oeis.org/A080277}} by proving that,
for all $n$, the sum of the first $n$ values of $\xi$ is asymptotic to $n\log_2 n$.

\begin{lemma}
\label{lem:brockhaus}
Let $\zeta_n=\sum_{i=1}^{n}\xi_n$.
Then $n\log_2 n - 2n< \zeta_n\le n\log_2 n+n$.
\end{lemma}

\smallskip\noindent{\bf Proof:}
It follows from the recursive construction of sequence $\xi$ that, when $n$ is a power of two,
\begin{align*}
\zeta_n & = 2 \zeta_{n/2} - (n - 1) + (2n - 1) = 2 \zeta_{n/2} + n = n (\log_2 n + 1).
\end{align*}
More generally, let the binary representation of $n$ be $n=\sum_{i=0}^k b_i2^i$ for $b_i\in\{0,1\}$, where $k=\lfloor\log_2 n\rfloor$. Then combining the evaluation of $\zeta_n$ at powers of two with the recursive construction of sequence $\xi$ gives a formula for $\zeta_n$ in terms of the binary representation of~$n$,
\begin{equation*}
\zeta_n = \sum_{i=0}^k \zeta_{2^i} b_i = \sum_{i=0}^k 2^i (i+1) b_i,
\end{equation*}
from which it follows that
\begin{equation*}
\zeta_n\le\sum_{i=0}^k (\log_2 n + 1) b_i 2^i = n\log_2 n + n.
\end{equation*}
In the other direction,
\begin{align*}
\zeta_n &=\sum_{i=0}^k (i+1)b_i2^i =\sum_{i=0}^k \left((k+1)b_i-(k-i)b_i\right)2^i \\
&\ge \sum_{i=0}^k \left((k+1)b_i-(k-i)\right)2^i = n (\lfloor\log_2 n\rfloor + 1) - 2^k\sum_{i=0}^k (k-i)2^{i-k}\\
&> n\log_2 n - n\sum_{j=0}^\infty\frac{j}{2^j}= n\log_2 n - 2n.\makebox[0pt]{\hspace{28.5em}$\Box$}
\end{align*}

\begin{lemma}
\label{lem:sawtooth}
Let the finite sequence $\alpha_1,\alpha_2,\dots \alpha_k$ of positive integers have sum $n$. Then there is a subsequence $\beta_1,\beta_2,\dots \beta_k$ of the first $n$ terms of $\xi$ such that, for all $i$, $\alpha_i\le \beta_i$.
\end{lemma}

\begin{proof}
We use induction on $n$.
Let $q$ be the largest power of two that is less than or equal to $n$; then
$\xi_q=2q-1\ge n\ge\max_i\alpha_i$.
Let $i$ be the smallest index for which the first $i$ values of $\alpha$ sum to at least $q$.
We choose $\beta_i=\xi_q$ and continue recursively for the two subsequences of $\alpha$ on either side of $\alpha_i$ and the two subsequences of $\xi$ on either side of $\xi_q$.
The sum of each of the two subsequences of $\alpha$ is strictly less than $q$ and each of the two subsequences of $\xi$ coincides with $\xi$ itself for $q-1$ terms, so by the induction hypothesis each of these two subproblems has a solution that can be combined with the choice $\beta_i=q$ to form a solution for the given input.
\end{proof}

\section{From superpatterns to universal point sets}
\label{sec:superpattern-to-universal}

In this section, we show how 213-avoiding superpatterns can be turned into universal point sets for planar graphs. Let $G$ be a planar graph. We assume $G$ is {\em maximal planar}, meaning that no additional edges can be added to $G$ without breaking its planarity; this is without loss of generality, because a point set that is universal for maximal planar graphs is universal for all planar graphs. Additionally, we assume that $G$ has a fixed plane embedding; for maximal planar graphs, such an embedding is determined by the choice of which of the triangles of $G$ is to be the outer face, and by the orientation of the outer triangle. With this choice fixed, we say that $G$ is a \emph{maximal plane graph}.

\subsection{Canonical representation}

As in the grid drawing method of de Fraysseix, Pach and Pollack~\cite{FraPacPol-IJC-1990}, we use \emph{canonical representations} of planar graphs:

\begin{definition}
A \emph{canonical representation} of a maximal plane graph $G$ is a sequence $v_1, v_2, \dots v_n$ of the vertices of $G$ such that $G$ can be embedded with the following three properties:
\begin{itemize}
\item $v_1v_2v_n$ is the outer triangle of the embedding, and is embedded in the clockwise order $v_1$, $v_n$, $v_2$.
\item For $3 \le k \le n$, the subgraph $G_k$ induced by $\{ v_1, \dots, v_k \}$ is $2$-connected and the boundary $C_k$ of its induced embedding is a cycle containing the edge $v_1 v_2$.
\item For $4 \le k \le n$, $v_k$ is on the outer face of $G_{k-1}$ and its earlier neighbors induce a path in $C_{k-1} \setminus v_1 v_2$ with at least two elements.
\end{itemize}

\end{definition}
As de Fraysseix, Pach and Pollack proved, every embedded maximal planar graph has at least one canonical representation. For the rest of this section, we will assume that the vertices $v_i$ of the given maximal plane graph $G$ are numbered according to such a representation.

\begin{definition}
Given a canonically represented graph $G$, let
\[
\parent(v_2) = \parent(v_3) = v_1,
\]
and for $i \ge 4$, let $\parent(v_i)$ be the neighbor of $v_i$ closest to $v_1$ on $C_{i-1} \setminus v_1v_2$.
By following a path of edges from vertices to their parents, every vertex can reach $v_1$, so these edges form a tree $\ctree(G)$ having $v_1$ as its root.
\end{definition}

The same tree $\ctree(G)$ may also be obtained by orienting each edge of $G$ from lower to higher numbered vertices, and then performing a depth-first search of the resulting oriented graph that visits the children of each vertex in clockwise order, starting from~$v_1$. Although we do not use this fact, $\ctree(G)$ is also part of a Schnyder decomposition of $G$, the other two trees of which are a second tree rooted at~$v_2$ that connects each vertex to its most counterclockwise earlier neighbor and a third tree rooted at~$v_n$ that connects each vertex to the later vertex whose addition removes it from~$C_k$.

\begin{definition}
For each vertex $v_i$ of $G$, let $\pre(v_i)$ be the position of $v_i$ in a pre-order traversal of $\ctree(G)$ that visits the children of each node in clockwise order, and let $\rpost(v_i)$ be the position of~$v_i$ in a sequence of the nodes of $\ctree(G)$ formed by reversing a post-order clockwise traversal.
\end{definition}
See Figure~\ref{fig:canon-perm} for an example. Note that $\rpost(v_i)$ is also the position of $v_i$ in a pre-order traversal in counter-clockwise order.

These two numbers may be used to determine the ancestor-descendant relationships in $\ctree(G)$: a node $v_i$ is an ancestor of a node $v_j$ if and only if both $\pre(v_i)<\pre(v_j)$ and $\rpost(v_i)<\rpost(v_j)$~\cite{Die-STOC-82}.

\begin{figure}[ht]
  \begin{minipage}[t]{0.47\linewidth}
    \centering
    \includegraphics[width=\textwidth]{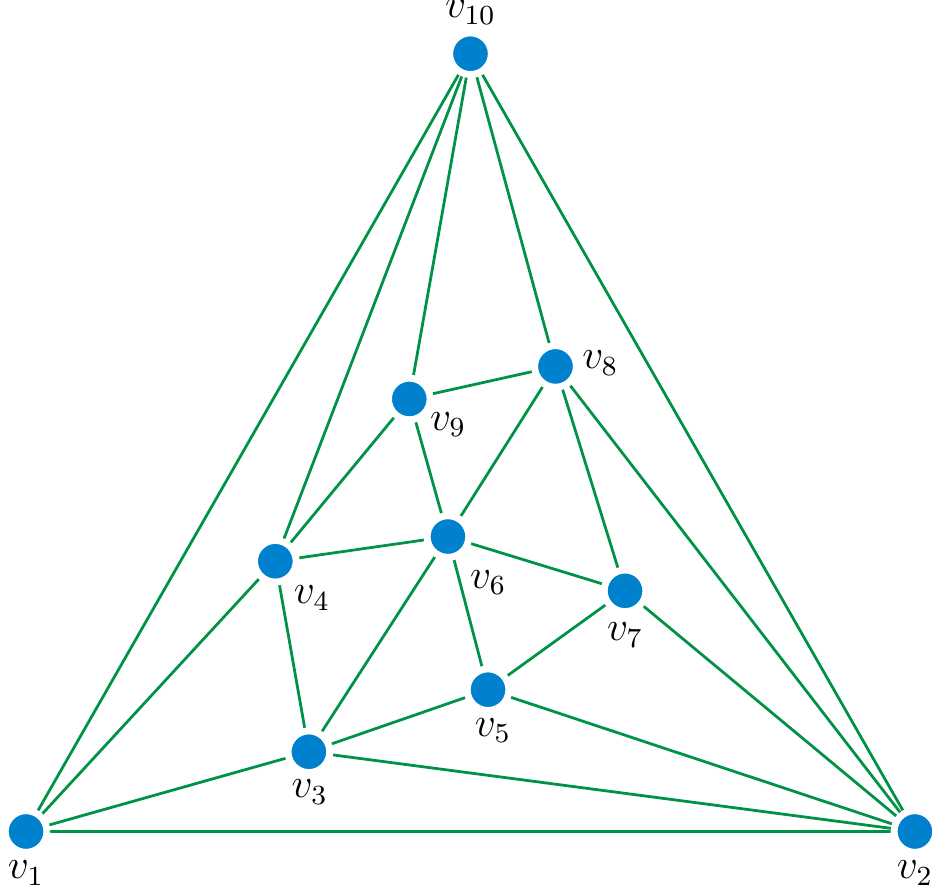}
  \end{minipage}\hfill
  \begin{minipage}[t]{0.47\linewidth}
    \centering
    \includegraphics[width=\textwidth]{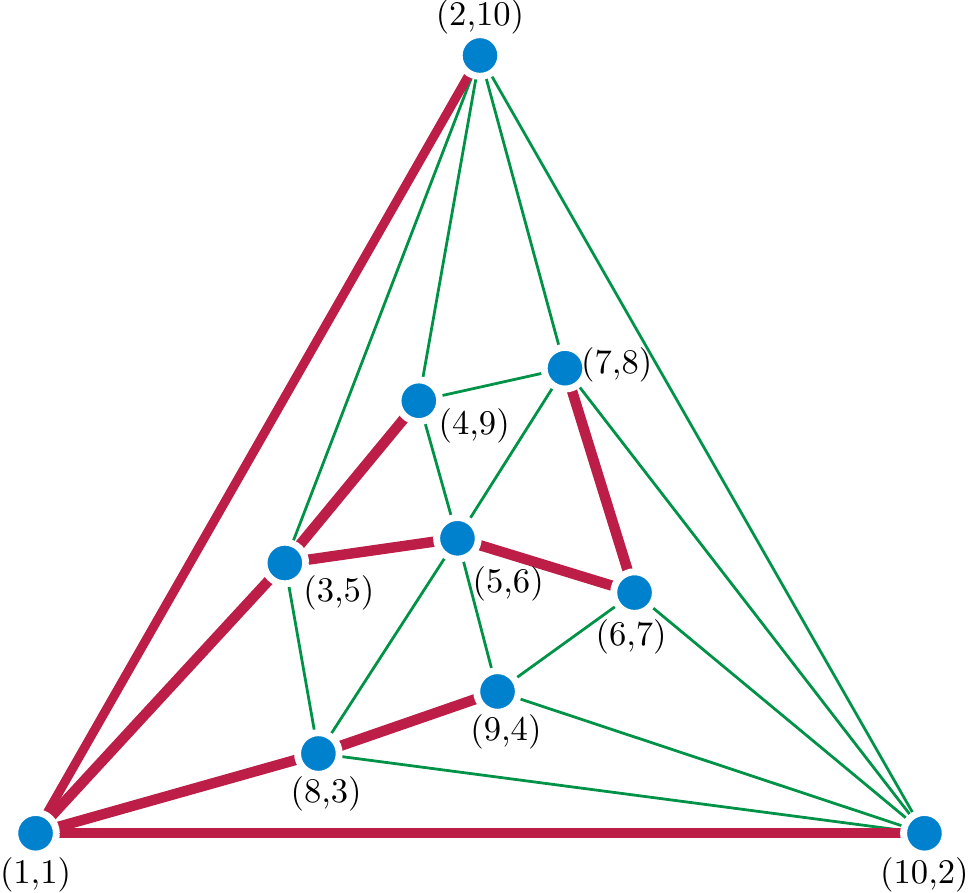}
  \end{minipage}
\caption{\label{fig:canon-perm} Left: a maximally planar graph $G$ with canonically ordered vertices. Right: $\ctree(G)$ is shown in red, and the label around each vertex $v$ indicates $\pre(v)$ and $\rpost(v)$.}
\end{figure}

\begin{lemma}\label{lem:C_k_sorted}
For every $3 \le k \le n$, the clockwise ordering of the vertices along the cycle $C_k$ is in sorted order by the values of $\pre(v_i)$.
\end{lemma}

\begin{proof}
The statement on the ordering of the vertices of $C_k$ follows by induction, from the fact that $v_k$ has a larger value of $\pre(\cdot)$ than its earliest incoming neighbor (its parent in $\ctree(G)$) and a smaller value than all of its other incoming neighbors.
\end{proof}

\begin{lemma}\label{lem:recanonize}
Let $G$ be a maximal plane graph together with a canonical representation $v_1, \dots, v_n$, and renumber the vertices of~$G$ in order by their values of $\rpost(v_i)$. Then the result is again a canonical representation of the same embedding of $G$.
\end{lemma}

\begin{proof}
The fact that $\rpost(v_i)$ gives a canonical representation comes from the fact that it is a reverse postorder traversal of a depth-first search tree. A reverse postorder traversal gives a topological ordering of every directed acyclic graph~\cite{Tar-AI-76}, from which it follows that every vertex in $G$ has the same set of earlier neighbors when ordered by $\rpost(v_i)$ as it did in the original ordering.
\end{proof}

\begin{lemma}\label{lem:midpre-interior}
Let $v_1, \dots, v_n$ be the canonical representation given by the $\rpost$ values of a maximal plane graph $G$. Suppose $\pre(v_h) < \pre(v_i) < \pre(v_j)$, $v_h$ and $v_j$ are neighbors, and $l = \max\{h, j\} > i$, then $v_i$ is not on $C_l$.
\end{lemma}

\begin{proof}
Suppose $h > j$. Then by Lemma~\ref{lem:C_k_sorted}, the vertices along $C_{h-1}$ and $C_h$ are sorted by their $\pre$ values. If $v_i$ is not on $C_{h-1}$, then certainly it is not on $C_h$. If $v_i$ is on $C_{h-1}$, then $v_h$ must be adjacent to a vertex on $C_{h-1}$ earlier than $v_i$ in order to maintain the sorted $\pre$ values on $C_h$. Moreover, $v_h$ is adjacent to $v_j$, which must appear on $C_{h-1}$ later than $v_i$. Therefore, $v_i$ must be in the interior of $G_h$ and not on $C_h$. The case $j < h$ is proved analogously.
\end{proof}

\begin{definition}
Let $\cperm(G)$ be the permutation in which, for each vertex $v_i$, the permutation value in position $\pre(v_i)$ is $\rpost(v_i)$. That is, $\cperm(G)$ is the permutation given by traversing $\ctree(G)$ in preorder and listing for each vertex of the traversal the number $\rpost(v_i)$.
\end{definition}

\begin{lemma}\label{lem:pi_G_213_avoiding}
For every canonically-represented maximal planar graph $G$, the permutation
$\pi = \cperm(G)$ is $213$-avoiding.
\end{lemma}

\begin{proof}
Let $i<j<k$ be an arbitrary triple of indices in the range from $1$ to $n$, corresponding to the vertices $u_i$, $u_j$ and $u_k$. Recall that a vertex $a$ is the ancestor of $b$ in $\ctree(G)$ if and only if $\pre(a) < \pre(b)$ and $\rpost(a) < \rpost(b)$.
If $\pi_j$ is not the smallest of these three values, then $\pi_i$, $\pi_j$, and $\pi_k$ certainly do not form a $213$ pattern. If $\pi_j$ is the smallest of these three values, then $u_i$ is not an ancestor or descendant of $u_j$, and $u_j$ is an ancestor of $u_k$. Therefore $u_i$ is also not an ancestor or descendant of $u_k$, from which it follows that $\pi_i>\pi_k$ and the pattern formed by $\pi_i$, $\pi_j$, and $\pi_k$ is $312$ rather than $213$. Since the choice of indices was arbitrary, no three indices can form a $213$ pattern and $\pi$ is $213$-avoiding.
\end{proof}

We observe that $\cperm(G)$ has some additional structure, as well: its first element is $1$, its second element is $n$, and its last element is~$2$.

\subsection{Stretching a permutation}

It is natural to represent a permutation $\sigma$ by the points with Cartesian coordinates $(i,\sigma_i)$, but for our purposes we need to stretch this representation in the vertical direction; we use a transformation closely related to one used by Bukh, Matou\v{s}ek, and Nivasch~\cite{BukMatNiv-IJM-11} for weak epsilon-nets, and by Fulek and T\'oth~\cite{FulTot-WADS-2013} for universal point sets for plane 3-trees.

\begin{definition}
Letting $q=|\sigma|$, we define
$$\stretchperm(\sigma)=\bigl\{(i,q^{\sigma_i})\mid 1\le i\le q\bigr\}.$$
\end{definition}

Let $\sigma$ be an arbitrary permutation with $q = |\sigma|$, and let $p_i$ denote the point in $\stretchperm(\sigma)$ corresponding to position $i$ in $\sigma$.

\begin{lemma}
\label{lem:slope}
Let $i$ and $j$ be two indices with $\sigma_i<\sigma_j$, and let $m$ be the absolute value of the slope of line segment $p_ip_j$. Then $q^{\sigma_j-1}\le m<q^{\sigma_j}$.
\end{lemma}

\begin{proof}
The minimum value of $m$ is obtained when $|i-j|=q-1$ and $\sigma_i=\sigma_j-1$, for which $q^{\sigma_j-1}=m$. The maximum value of $m$ is obtained when $|i-j|=1$ and $\sigma_i=1$,
for which $m=q^{\sigma_j}-q<q^{\sigma_j}$.
\end{proof}

\begin{lemma}
\label{lem:tri-orient}
Let $i$, $j$, and $k$ be three indices with $\max\{\sigma_i,\sigma_j\}<\sigma_k$ and $i<j$. Then the clockwise ordering of the three points $p_i$, $p_j$, and $p_k$ is $p_i$, $p_k$, $p_j$.
\end{lemma}

\begin{proof}
The result follows by using Lemma~\ref{lem:slope} to compare the slopes of the two line segments $p_ip_j$ and $p_ip_k$.
\end{proof}

\begin{lemma}
\label{lem:crossing-order}
Let $h$, $i$, $j$, and $k$ be four indices with $\max\{\sigma_h, \sigma_i,\sigma_j\}<\sigma_k$ and $h<j$. Then line segments $p_hp_j$ and $p_ip_k$ cross if and only if both $h<i<j$ and $\max\{\sigma_h,\sigma_j\}>\sigma_i$.
\end{lemma}

\begin{proof}
A crossing occurs between two line segments if and only if the endpoints of every segment are on opposite sides of the line through the other segment. The endpoints of $p_ip_k$ are on opposite sides of line $p_hp_j$ if and only if the two triangles $p_hp_ip_j$ and $p_hp_kp_j$ have opposite orientations; analogously, the endpoints of $p_hp_j$ are on opposite sides of line $p_ip_k$ if and only if the two triangles $p_ip_hp_k$ and $p_ip_jp_k$ have opposite orientations. Note that by Lemma~\ref{lem:tri-orient} and the assumption that $\sigma_k$ is the largest, $p_hp_ip_j$ and $p_hp_kp_j$ having opposite orientations implies that $\sigma_i$ is not the second largest, and $p_ip_hp_k$ and $p_ip_jp_k$ having opposite orientations implies that $h < i < j$. Therefore, if the line segments cross, then the two conditions $h<i<j$ and $\max\{\sigma_h,\sigma_j\}>\sigma_i$ are satisfied. Conversely, if the two conditions are satisfied, then the same lemma implies that the two pairs of triangles have opposite orientations, and the line segments must cross.
\end{proof}

\subsection{Universal point sets}

\begin{definition}
If $\sigma$ is any permutation, we define $\augment(\sigma)$ to be a permutation of length
$|\sigma|+3$, in which the first element is $1$, the second element is $|\sigma|+3$, the last element is $2$, and the remaining elements form a pattern of type $\sigma$.
\end{definition}

It follows from Lemma~\ref{lem:pi_G_213_avoiding} that, if $\sigma$ is an $S_{n-3}(213)$-superpattern and if $G$ is an arbitrary $n$-vertex maximal plane graph, then $\cperm(G)$ is a pattern in $\augment(\sigma)$.

\begin{theorem}
\label{thm:U_sigma_universal}
Let $\sigma$ be an $S_{n-3}(213)$-superpattern, and let \[U_n=\stretchperm(\augment(\sigma)).\]
Then $U_n$ is a universal point set  for planar graphs on $n$ vertices.
\end{theorem}

\begin{proof}
Let $G$ be an $n$-vertex maximal plane graph and $v_1,v_2,\dots v_n$ be the canonical representation of $G$ given by the $\rpost$ values. Let $x_i$ denote a sequence of positions in $\augment(\sigma)$ that
form a pattern of type $\cperm(G)$, with position $x_i$ in $\augment(\sigma)$ corresponding to position $\pre(v_i)$ in $\cperm(G)$. Let $q=|\augment(\sigma)|$, and for each $i$, let $y_i=q^j$ where $j$ is the value of $\augment(\sigma)$ at position $x_i$.
Embed $G$ by placing vertex $v_i$ at the point $(x_i,y_i)\in U_n$.

Let $v_h v_j$ and $v_i v_k$ be two edges of $G$, where we assume without loss of generality that $\rpost(v_k)$ is larger than the $\rpost$ value of the other three vertices. If these two edges crossed in the embedding of $G$, then by Lemma~\ref{lem:crossing-order} we would necessarily have $\pre(v_h)<\pre(v_i)<\pre(v_j)$, and $\rpost(v_i)<\max\{\rpost(v_h),\rpost(v_j)\}$. 
By Lemma~\ref{lem:midpre-interior}, $v_i$ would not be on the outside face of the graph induced by the vertices with $\rpost$ values at most $\max\{\rpost(v_h),\rpost(v_j)\}$, and could not be a neighbor of $v_k$. This contradiction shows that no crossing is possible, so the embedding is planar.
\end{proof}

\section{Superpatterns for $S_n(213)$}
In this section we construct a $S_n(213)$-superpattern of size $n^2/4 + n + ((-1)^n - 1)/8$ and then use these superpatterns to find small superpatterns for the family of all $n$-vertex planar graphs.

We verified using computer searches that this size is minimal for $n \leq 6$. For $n\le 5$, this verification was done by exhaustively searching all permutations up to the given size. For $n=6$, for which the optimal superpattern size is~15, we used a slightly more sophisticated search strategy: we exhaustively generated all length-$13$ permutations, formed the subset of length-$13$ $S_5(213)$-superpatterns, and then verified for each of these that there was no way of adding one more element to form an $S_6(213)$-superpattern.

Our construction begins with a lemma demonstrating the recursive structure of $S_n(213)$-super\-patterns that have exactly $n$ rows and $n$ columns. The $S_n(213)$-superpattern that we construct will have this property. Note that an $S_n(213)$-superpattern must have at least $n$ rows and $n$ columns in order to embed the pattern consisting of a single decreasing run.
\begin{lemma}\label{lem:row-col-prop}
If $\sigma$ is a $S_n(213)$-superpattern and has $n$ rows and $n$ columns, then the permutation described by the intersection of columns  $n-j + 1$ to $n-i + 1$  and rows $i$ to $j$  of $\sigma$ is a $S_{j-i+1}(213)$-superpattern, for every $1 \le i \le j \le n$.
\end{lemma}
\begin{proof}
Let $\pi$ be an arbitrary $213$-avoiding permutation of length $j-i+1$ and consider the
$n$-element $213$-avoiding permutation
$$\tau=n(n-1)\dots(j+1)(\pi_1+i-1)(\pi_2+i-1)\dots(\pi_{j-i+1}+i-1)(i-1)(i-2)\dots 321.$$
(See Fig.~\ref{fig:sigmahat}.) By the assumption that $\sigma$ is a superpattern, $\tau$ has an embedding into $\sigma$. Because there are $n-j$ descents in $\tau$ before the first element of the form $\pi_k+i-1$, this embedding cannot place any element $\pi_k+i-1$ into the first $n-j$ columns of $\sigma$.  Similarly because there are $i$ descents in $\tau$ after the last element of the form $\pi_k+i-1$,
this embedding cannot place any element $\pi_k+i-1$ into the last $i-1$ columns of $\sigma$.  By a symmetric argument, the elements of the form $\pi_k+i-1$ cannot be embedded into the $i-1$ lowest rows nor the $n-j$ highest rows of $\sigma$.  Therefore these elements, which form a pattern of type $\pi$,  must be embedded into $\sigma$ between column $n-j+1$ and column $n-i+1$ (inclusive) and between row $i$ and row $j$ (inclusive). Since $\pi$ was arbitrary, this part of $\sigma$ must be universal for permutations of length $j-i+1$, as claimed.
\end{proof}

\begin{figure}[t]
  \centering
  \includegraphics[height=1.75in]{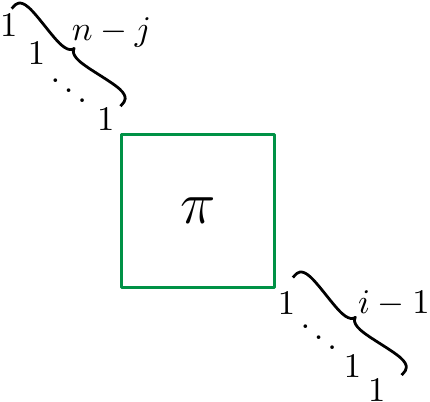}
  \caption{\label{fig:sigmahat} The permutation $\tau$ constructed from $\pi$ in the proof of Lemma~\ref{lem:row-col-prop}}
\end{figure}

\begin{figure}[t]
\centering
\includegraphics[width=0.95\textwidth]{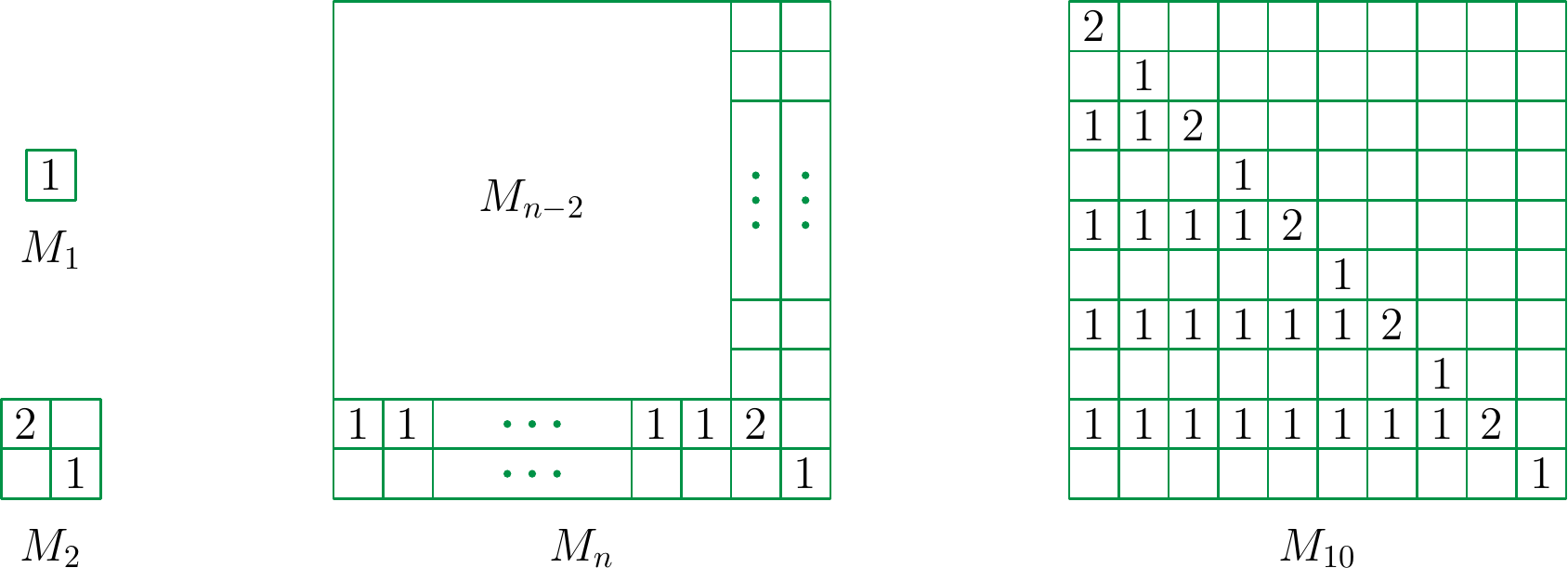}
\caption{
    \label{fig:m_n}
    The base cases, $chessboard(\mu_i)$ for $i = 1,2$, and the inductive construction of $\chessboard(\mu_n)$ from $\chessboard(\mu_{n-2})$. Cells of the matrices containing zero are shown as blank.
}
\end{figure}

\begin{figure}[t]
\centering
\includegraphics[width=0.85\textwidth]{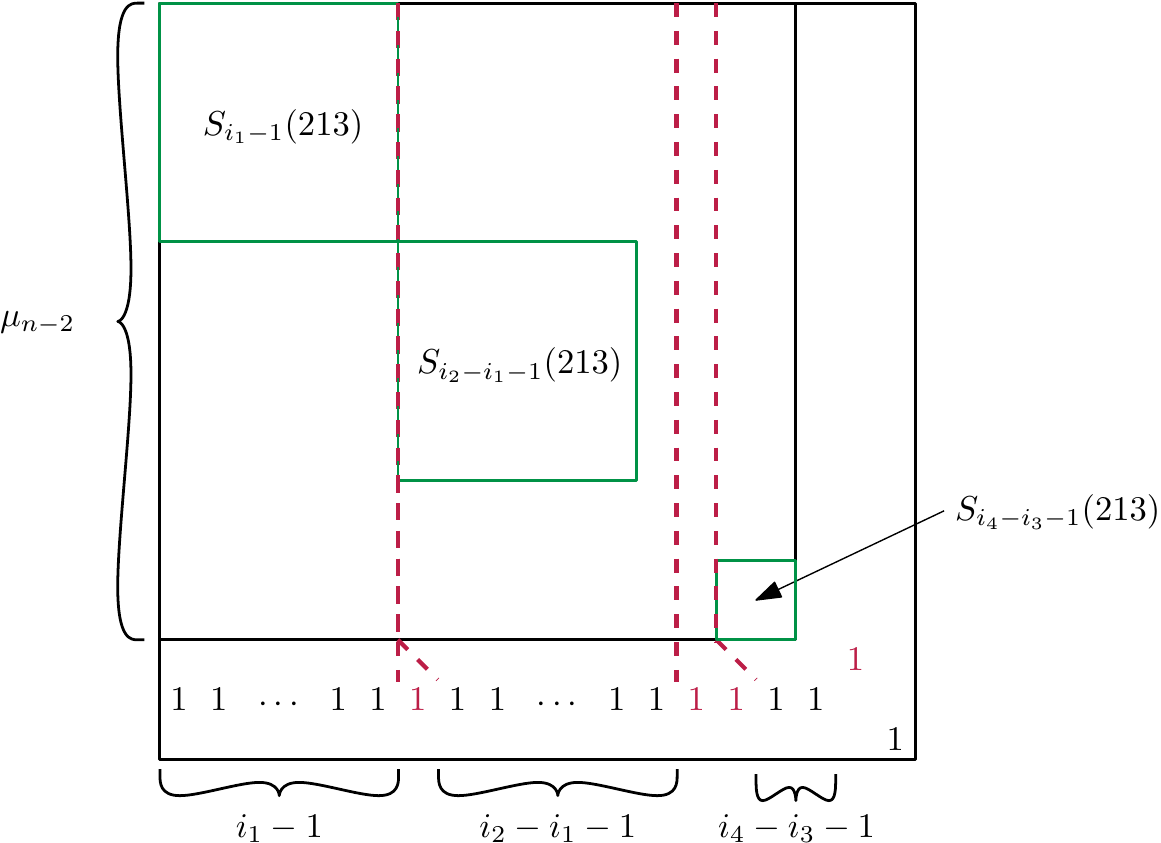}
\caption{\label{fig:embed-examp} A partial embedding of the red elements, showing where the remaining blocks can be fit into the columns of $\mu_{n-2}$. The run of length $2$ in the bottom right of the chessboard has been expanded for clarity.}
\end{figure}

We define a permutation $\mu_n$, which we will  show to be a $S_n(213)$-super\-pattern, by describing $\chessboard(\mu_n)=M_n$. Recall that in a chessboard representation $M$, $M(i,j)$ denotes the cell at the $i$th column (from left to right) and $j$th row (from bottom to top) of the chessboard. In our construction, $M_n$ will have exactly $n$ columns and $n$ rows. The bottom two rows of $M_n$ will have $M_n(n,1) = 1$, $M_n(i,2) = 1$ for all $1 \leq i \leq n-2$, $M_n(n-1,2) = 2$, and all other values in these two rows will be zero. The values in the top $n-2$ rows are given recursively by $M_n{(1\colon\! n-2, 3\colon\! n)}=M_{n-2}$, again with all values outside this submatrix being zero. The base cases of $\mu_1$ and $\mu_2$ and the inductive definition with an example are shown in Figure~\ref{fig:m_n}.

\begin{theorem}\label{thm:menorah_universal}
The permutation $\mu_n$ is a $S_n(213)$-superpattern. Thus there exists a $S_n(213)$-superpattern whose size is $n^2/4 + n +((-1)^n - 1) / 8$.
\end{theorem}
\begin{proof}
It can be easily verified that $\mu_n$ is a $S_n(213)$-superpattern when $1 \leq i \leq 2$. Let $\pi$ be an arbitrary $213$-avoiding permutation of length $n > 2$.  We will show that $\pi$ can be embedded into $\mu_n$. We have two cases, depending on the value of the last element of $\pi$.

\begin{description}
\item[Case 1:] $\pi_n = 1$

Let $\pi_{i_1} \dots \pi_{i_k}$ be the second lowest row of $\pi$.  Observe that $i_k = n - 1$ and $\pi_{i_1} = 2$.  We claim that $\pi_{i_j} = j+1$ for all $2 \le j \le k$.  Indeed, in one direction, $\pi_{i_j} \ge j+1$, since $\pi_{i_j}$ is larger than all $j-1$ preceding elements on the second lowest row, and every such preceding element has value at least $2$.  In the other direction, $\pi_{i_j} \le j+1$, for otherwise, there is some $l < i_{j-1}$ such that $\pi_l = j+1$, forming a $213$-pattern with $\pi_l \pi_{i_{j-1}} \pi_{i_j}$.

In this case, we embed the bottom two rows of $\pi$ by mapping $\pi_n$ to the bottom right element of $\mu_n$ and $\pi_{i_j}$ to the $i_j$-th position of the second lowest row of $\mu_n$.

\item[Case 2:] $\pi_n \neq 1$

Let $\pi_{i_1} \dots \pi_{i_k}$ be the lowest row of $\pi$.  Similarly to Case~1, $i_k = n$, and because $\pi$ is $213$-avoiding, $\pi_{i_j} = j$ for all $1 \le j \le n$.  We embed this bottom row of $\pi$ by mapping $\pi_{i_j}$ to the $i_j$-th position of the second lowest row of $\mu_n$ (in the case $k = n$, the $i_{n-1}$-th and $i_n$-th positions are both at the cell $M_n({n-1},2)$ of $\mu_n$).
\end{description}

To finish the embedding, the remaining elements need to be embedded into the copy of $\mu_{n-2}$.  Recall that a block of a permutation is a contiguous subsequence formed by a set of consecutive integers.  Because $\pi$ is $213$-avoiding, the remaining elements of $\pi$ form disjoint blocks that fit in the columns between the elements embedded so far. If one block is to the right of another in $\pi$, then every element in that block has a smaller value than every element in the block to the left. Let $\pi_{i_1}$ be the leftmost element that has been embedded on the second row of $\mu_n$.  Then there are $i_1-1$ elements to the left of $\pi_{i_1}$ in $\pi$ and $i_1-1$ columns to the left of the column where $\pi_{i_1}$ was embedded in $\mu_n$.  By Lemma~\ref{lem:row-col-prop}, these elements can fit into the top $i_1-1$ rows of these columns.  Now let $\pi_{i_j}$ and $\pi_{i_{j+1}}$ be two adjacent elements embedded in the second lowest row of $\mu_n$.  Between these two there are $i_{j+1}-i_j-1$ columns of $\mu_{n-2}$ available: from the column above $\pi_{i_j}$ to the column before $\pi_{i_{j+1}}$.  So again by Lemma~\ref{lem:row-col-prop}, the $i_{j+1}-i_j-1$ elements between $\pi_{i_j}$ and $\pi_{i_{j+1}}$ can be fit into rows $n-i_j+1$ to $n-i_{j+1}$ of those columns (see Fig.~\ref{fig:embed-examp}).  Because $i_k = n-1$ in Case 1 and $n$ in Case 2, there is no block after $\pi_{i_k}$. Therefore $\pi$ can be embedded into $\mu_n$ and $\mu_n$ is a $S_n(213)$-superpattern.

It remains to compute the size of $\mu_n$. We have $|\mu_1| = 1$, $|\mu_2| = 3$ and from the recursive definition of $\mu_n$
\begin{align*}
    |\mu_n| &= (n+1) + |\mu_{n-2}| \\
            &= (n+1) + (n-2)^2/4 + (n-2) + ((-1)^{n-2} - 1)/8\\
            &= (n+1) + (n^2/4 - n + 1) + (n-2) + ((-1)^{n-2} - 1)/8\\
            &= n^2/4 + n + ((-1)^{n} - 1)/8,
\end{align*}
proving the size claimed in the theorem by induction on $n$.
\end{proof}

Combining Theorem~\ref{thm:menorah_universal} with Theorem~\ref{thm:U_sigma_universal}, the following is immediate:

\begin{theorem}
The $n$-vertex planar graphs have universal point sets of size $n^2/4 - \Theta(n)$.
\end{theorem}

\section{Specific subclasses  of 213-avoiding permutations}

In this section we investigate the size of superpatterns for certain permutation classes (defined by forbidden patterns) that are strict subsets of
the $213$-avoiding permutations. For each of these permutation classes, we prove that there exist superpatterns of linear or near-linear size. Each of these permutation classes corresponds to a subfamily of the planar graphs (although not necessarily a natural subfamily): the graphs for which the permutation given by our algorithm to map maximal plane graphs to permutations happens to be in the given class.
It follows by the same stretching construction as for our main result that these subfamilies of planar graphs have universal point sets of near-linear size.

Our results in this section include the strongest lower bound we currently know on the size of $S_n(213)$-superpatterns. Additionally, the upper bounds in this section will serve as a warm-up for the more general upper bounds in Section~\ref{sec:general-subclasses} for all subclasses of $213$-avoiding permutations.

\subsection{Superpatterns for $S_n(213, 312)$}

\begin{figure}[b]
\centering
\includegraphics[width=0.9\textwidth]{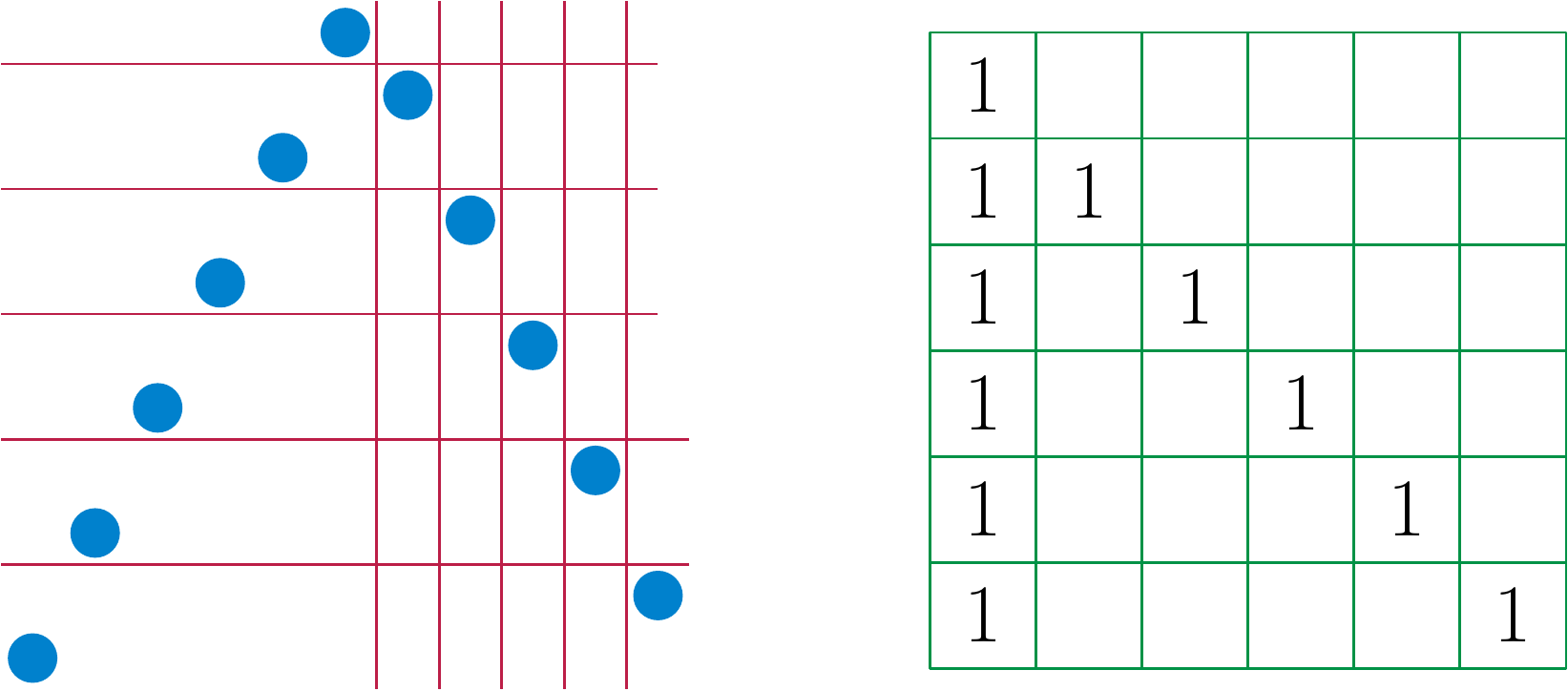}
\caption{The $S_6(213,312)$-superpattern}
\end{figure}

A permutation is in $S_n(213,312)$ if and only if it is \emph{unimodal}: all its ascents must occur earlier than all of its descents. There are exactly $2^{n-1}$ such permutations: each permutation in this family of permutations is determined by the set of elements that are earlier than the largest element~\cite{Sim-Sch-85}. A superpattern $\pi$ for $S_n(213,312)$ must have at least $2n-1$ elements, because it must allow the two permutations $123\ldots n$ and $n(n-1)(n-2)\ldots 1$ to both be embedded into $\pi$, and their embeddings can only share a single element. This bound is tight:

\begin{theorem}
$S_n(213,312)$ has a minimal superpattern of length $2n-1$.
\end{theorem}

\begin{proof}
The permutation
\[ 135\ldots (2n-1)(2n-2)(2n-4)\ldots 642 \]
is a superpattern for $S_n(213,312)$ and has length exactly $2n-1$.
\end{proof}

\subsection{Superpatterns for $S_n(213, 132)$}
\label{sec:213-132}

A permutation is in $S_n(213,132)$ if and only if its chessboard notation has nonzero entries only on the diagonal from upper left to lower right of the chessboard; that is, if it is a descending sequence of ascending subsequences~\cite{Sim-Sch-85}. For instance, $789564123$ is a descending sequence of the four ascending subsequences $789$, $56$, $4$, and $123$; its chessboard notation has the lengths of these subsequences (3, 2, 1, 3) in the diagonal entries. Such a permutation is determined by a single bit of information for each pair of consecutive values: do they belong to the same ascending subsequence or not? Therefore, there are exactly $2^{n-1}$ such partitions~\cite{Rot-DM-81}, equinumerous with $S_n(213,312)$. However, as we will see, the minimum size of a superpattern for $S_n(213,132)$ is asymptotically different from the $2n-1$ bound on this size for $S_n(213,312)$.

\begin{theorem}
$S_n(213,132)$ has a superpattern of size at most $n\log_2 n+n$.
\end{theorem}

\begin{proof}
We form a superpattern that is itself $\{213,132\}$-avoiding, by constructing a permutation whose chessboard notation has the sequence $\xi_1,\xi_2,\dots,\xi_n$ on its main diagonal and zeros elsewhere (Figure~\ref{fig:213-132}).
The result follows from Lemma~\ref{lem:sawtooth} and from the bounds of Lemma~\ref{lem:brockhaus} for the partial sums of~$\xi$.
\end{proof}

\begin{figure}[ht]
\centering
\includegraphics[width=0.9\textwidth]{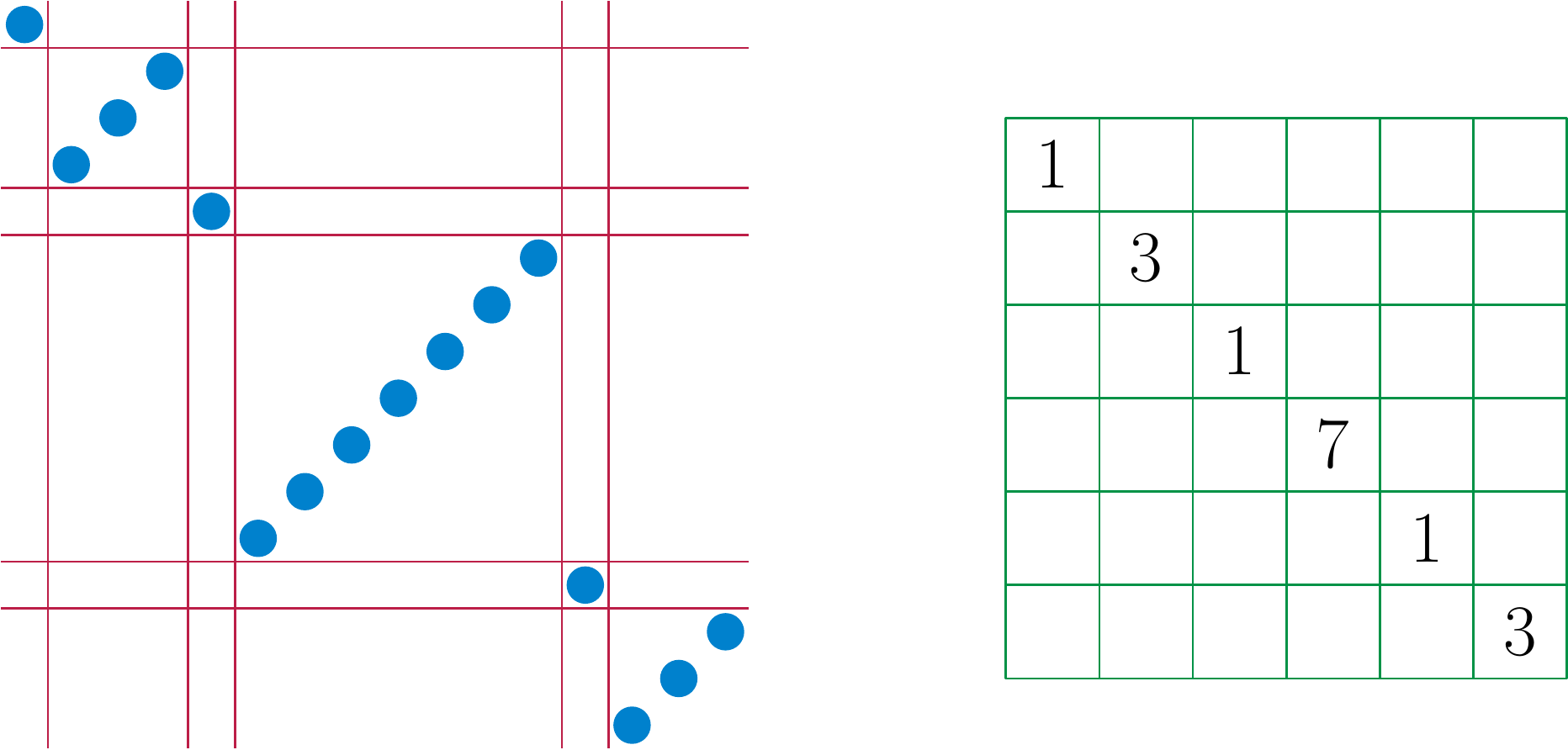}
\caption{The $S_7(213,132)$-superpattern}
\label{fig:213-132}
\end{figure}

\begin{theorem}
Every $S_n(213,132)$-superpattern has size $\Omega(n\log n)$.
\end{theorem}

\begin{proof}
For simplicity, we assume that $n$ is a power of two; the result extends to arbitrary $n$ by rounding $n$ down to the next smaller power of two.
Let $\pi$ be a $S_n(213,132)$-superpattern and
consider the patterns in $S_n(213,132)$ whose chessboard representations have on their diagonals the sequences $(n)$, $(n/2,n/2)$, $(n/4,n/4,n/4,n/4)$, etc. Observe that each sequence length and the values within the sequences are all powers of two, and that each sequence value is $n$ divided by the sequence length. Since $\pi$ is a superpattern, it must be possible to place each of these patterns somewhere in $\pi$; fix a choice of how to place each pattern. Observe that, in this placement, each ascending subsequence of one of these patterns can only have a nonempty intersection with at most one ascending subsequence of another.

Consider the values $i=0,1,2,3,\dots$ in ascending order by $i$.
The $i$th of the set of patterns described above consists of $2^i$ ascending sequences of length $n/2^i$. For $i=0$ we mark the single ascending sequence of length $n$.
For $i>0$, we mark $2^{i-1}$ of these ascending sequences, among the ones whose placement is disjoint from all previously marked sequences. This is always possible because there are exactly $1+\sum_{j=1}^{i-1}2^{j-1}=2^{i-1}$ previously marked sequences and each of them can prevent only one of the $2^i$ length-$(n/2^i)$ ascending subsequences from being marked.

The total number of distinct elements of $\pi$ in the marked subsequences is
\[n+\frac{n}{2}+2\,\frac{n}{4}+4\,\frac{n}{8}+\dots=\frac{n}{2}\log_2 n +n,\]
so $\pi$ must have at least that many elements in total.
\end{proof}

\begin{corollary}
Every superpattern for $S_n(213)$ must  have size $\Omega(n\log n)$.
\end{corollary}

\subsection{Superpatterns for $S_n(213, 3412)$}
\label{sec:213-3412}

Our next result depends on a structural characterization of the $213$-avoiding permutations in terms of their chessboard representations.

\begin{definition}
The \emph{chessboard graph} of a permutation is a directed acyclic graph that has a vertex for each nonzero square in the chessboard representation of the permutation, and that has an edge between each pair of nonzero squares if they belong to the same row or column of the representation and there is no nonzero square between them. These edges are oriented upwards for squares of the same column and rightwards for squares of the same row. 
\end{definition}

\begin{figure}[ht]
\centering
\includegraphics[width=0.9\textwidth]{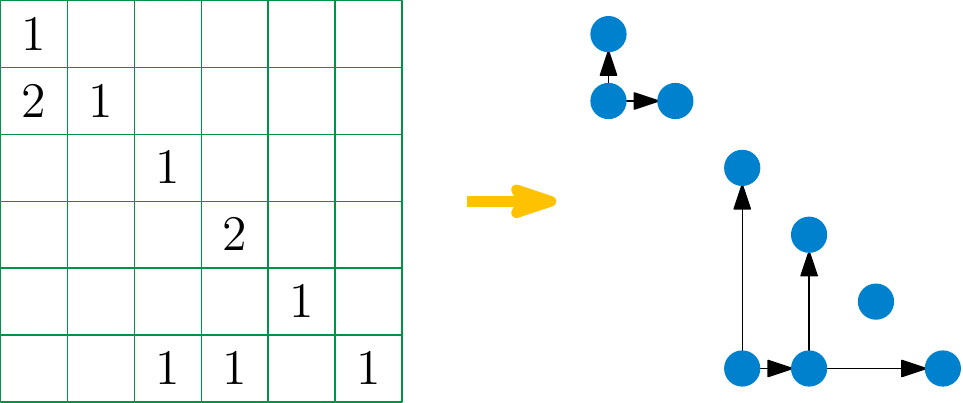}
\caption{The chessboard of a permutation and its corresponding chessboard graph.}
\label{fig:chessboard-graph}
\end{figure}

\begin{definition}
We say that a chessboard graph is a \emph{directed plane forest} if every vertex has at most one incoming neighbor, and the embedding of the graph with respect to the chessboard (i.e., embedding each vertex at the corresponding nonzero square of the chessboard) has no edge crossings.
\end{definition}

\begin{lemma} \label{lem:213-characterization}
A permutation $\pi$ is $213$-avoiding if and only if (1) its chessboard representation has the same number of rows and columns, (2) its chessboard representation has nonzeros in each square of the diagonal from top left to bottom right, and (3) its chessboard graph is a directed plane forest.
\end{lemma}

\begin{proof}
We first verify that a permutation meeting these conditions can have no $213$ pattern.
In a chessboard representation meeting these conditions, the squares above the diagonal must all be zero, for a minimal nonzero square above the diagonal would be the target of two incoming edges from the diagonal, violating the directed forest condition. Additionally, every vertex of the chessboard graph that is below the diagonal must have two outgoing edges, one upward and one rightward, because its row and column have at least one more nonzero each, on the diagonal itself. Suppose a $213$ pattern existed, then the nonzero squares of the chessboard corresponding to the $2$ and $1$ positions in this pattern would necessarily be to the left and below the nonzero square corresponding to the $3$ position. The nonzero square of the $3$ position is at or below the diagonal, so a nonzero square must exist on the same row as the $2$ position and on a column strictly to the left of the $1$ position whose rightward path reaches the column of the $3$ position, and a nonzero square must exist on the same column as the $1$ position and on a row strictly below the $2$ position whose upward path reaches the row of the $3$ position. These two paths must cross or meet at a nonzero square, violating the directed plane forest condition. Therefore, a $213$ pattern cannot exist.

It remains to show that a $213$-avoiding permutation necessarily has a chessboard representation that meets these conditions. Let $\pi$ be such a permutation, and let $L$ be the set of indexes of a longest decreasing subsequence of $\pi$. Among all such subsequences, choose $L$ to be maximal, in the sense that there is no index $i\in L$ and $i'\notin L$ with $i'>i$ and $L\setminus\{i\}\cup\{i'\}$ an equally long decreasing subsequence. With this choice, there can be no $i'>i$ with $\pi(i')>\pi(i)$, for such an~$i'$ would either violate the assumption of maximality or would form a $213$ pattern together with two members of $L$. Additionally, every column of $\pi$ must be represented in $L$, for if the maximum value in a column did not belong to $L$ then it, the next element in the permutation, and the next element in $L$ would together form a $213$ pattern. By a symmetric argument, every row of $\pi$ must be represented in $L$. Since $L$ is a decreasing sequence, each element of $L$ is in a distinct row and distinct column. Therefore, the numbers of rows and columns both equal $L$. Moreover, because $L$ is maximal, each diagonal square of the chessboard representation contains a member of $L$, so the diagonal is nonzero. The chessboard graph of $\pi$ must be a directed plane forest, because if not it would contain two edges that either meet or cross, and in either case the elements of the chessboard squares at the endpoints of these edges contain a $213$ pattern.
\end{proof}

The above proof also gives us the following:
\begin{lemma}
The chessboard representation of a $213$-avoiding permutation has zeros in every square above the diagonal.
\end{lemma}

We now characterize the chessboard representation and chessboard graph of $\{213,3412\}$-avoiding permutations. The characterization for the chessboard representation will be used later to give a superpattern for this class of permutations.

\begin{lemma} \label{lem:213-3412-characterization}
A permutation is $\{213,3412\}$-avoiding if and only if its chessboard representation is a square, where the diagonal elements are all ones except for a special square which may be greater than one, and the nonzero off-diagonal elements plus the special diagonal element (if it exists) form a sequence of squares whose two coordinates are monotonically increasing.
\end{lemma}

\begin{proof}
Let $\pi$ be a $\{213,3412\}$-avoiding permutation. By Lemma~\ref{lem:213-characterization}, its chessboard representation is square with nonzero diagonal elements. Call a diagonal element \emph{special} if it is greater than one. Consider the set $S$ consisting of the off-diagonal squares with nonzero elements and the special diagonal squares. Suppose there are two squares in $S$, one at column $i$ and one at column $j > i$, such that the square at column~$j$ is below the one at column~$i$. Then either we have a $213$ pattern formed by these two squares plus the diagonal square at column~$j$, or we have a $3412$ pattern formed by the two squares and the diagonal squares in their columns (note that one of the squares may be special and contributes two elements to the pattern). In either case, we get a contradiction, so $S$ must consist of a single sequence of squares with increasing coordinates. Since the second coordinates of diagonal elements are decreasing, $S$ contains at most one diagonal square. This shows that $\pi$ satisfies all the conditions in the lemma.

Conversely, let $\pi$ be a permutation whose chessboard representation satisfies the conditions in the lemma. By Lemma~\ref{lem:213-characterization}, $\pi$ has no $213$ patterns. Moreover, to embed a $3412$ pattern in $\pi$, the square corresponding to the $3$ must be an off-diagonal square or the special diagonal square with value greater than one. In either case, the only elements to the right and below the square corresponding to the $3$ are diagonal squares with value one, which can have no $12$ patterns. This shows that $\pi$ has no $3412$ patterns, and so is $\{213,3412\}$-avoiding. This completes the proof.
\end{proof}

\begin{lemma}
The chessboard graph of a $\{213,3412\}$-avoiding permutation is a disjoint union of caterpillars (trees in which there is a single path that contains all non-leaf nodes).
\end{lemma}

\begin{proof}
In the chessboard representation, the neighbors of each off-diagonal square are either diagonal squares or the next element in the increasing sequence of squares. The diagonals squares are leaves of the chessboard graph, so each nonleaf node can have at most one other nonleaf neighbor. This means that the nonleaf nodes induce paths in the forest, as needed.
\end{proof}

\begin{figure}[t]
\centering
\includegraphics[width=0.9\textwidth]{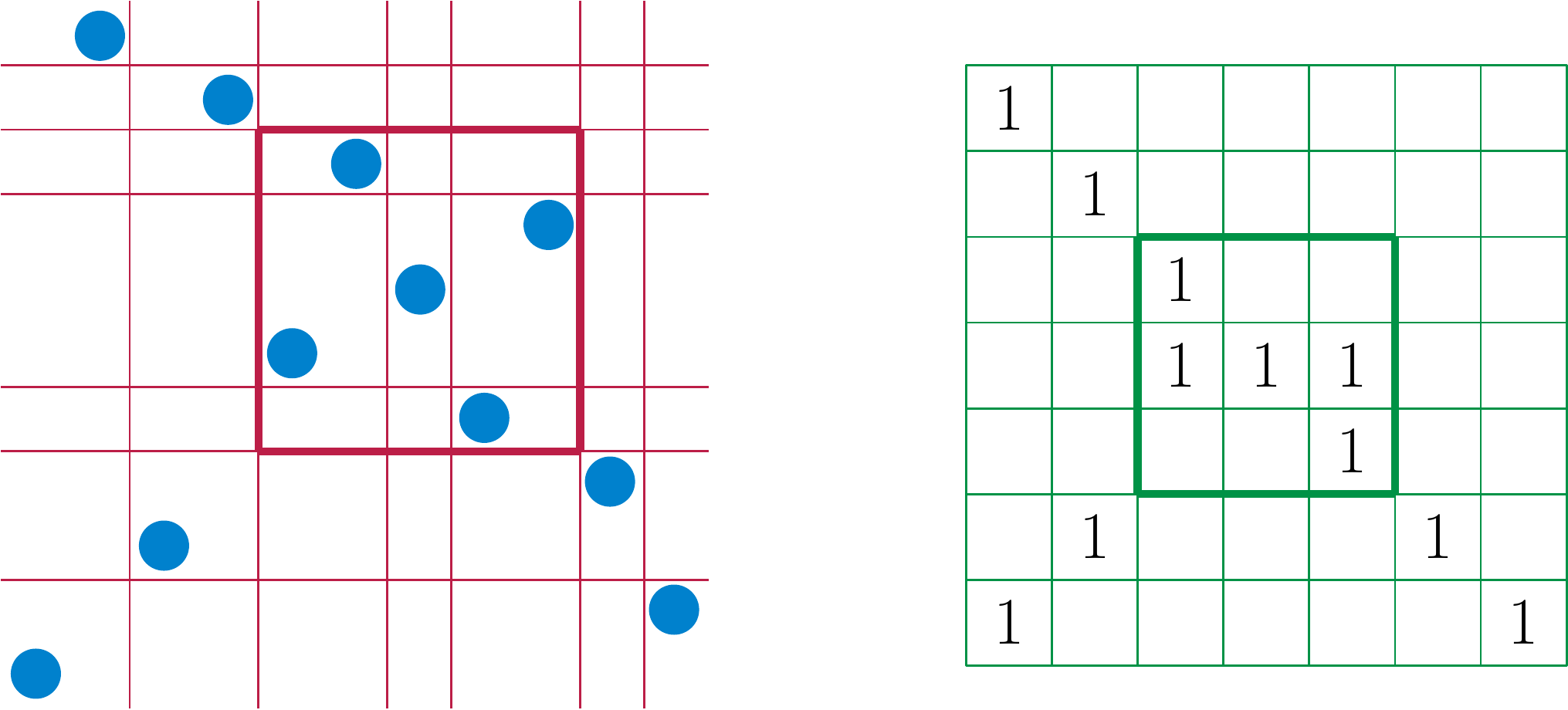}
\caption{The $S_5(213,3412)$-superpattern. The pattern inside the bolded box is the base case $S_3$-superpattern.}
\label{fig:213-3412}
\end{figure}

\begin{theorem}\label{thm:213-3412-superpattern}
$S_n(213, 3412)$ has a superpattern of length $3n-4$, for $n \ge 3$.
\end{theorem}

\begin{proof}
As a base case, $25314$ is a superpattern for $S_3(213, 3412)$, because it is a superpattern more generally for $S_3$.
For larger $n$, construct the chessboard representation for the superpattern $\sigma_n$ as follows: start with a $2n-3$ by $2n-3$ chessboard, place a copy of the permutation $25314$ in the central $3$ by $3$ portion of the grid, then put ones along the remaining main diagonal and the diagonal from the bottom left to the center of the grid (Figure~\ref{fig:213-3412}).
To embed a permutation $\pi \in S_n(213, 3412)$ into $\sigma_n$, for $n>3$, there are three cases based on the characterization in Lemma~\ref{lem:213-3412-characterization}: the smallest element of $\pi$ may be its first element, it may be its last element, or the largest element of $\pi$ may be the first element. In each case, the element matching the case may be covered by the bottom row or leftmost column of $\sigma_n$, and the result follows by induction.
\end{proof}

\begin{figure}[ht]
\centering
\subfloat{\includegraphics[width=0.25\textwidth]{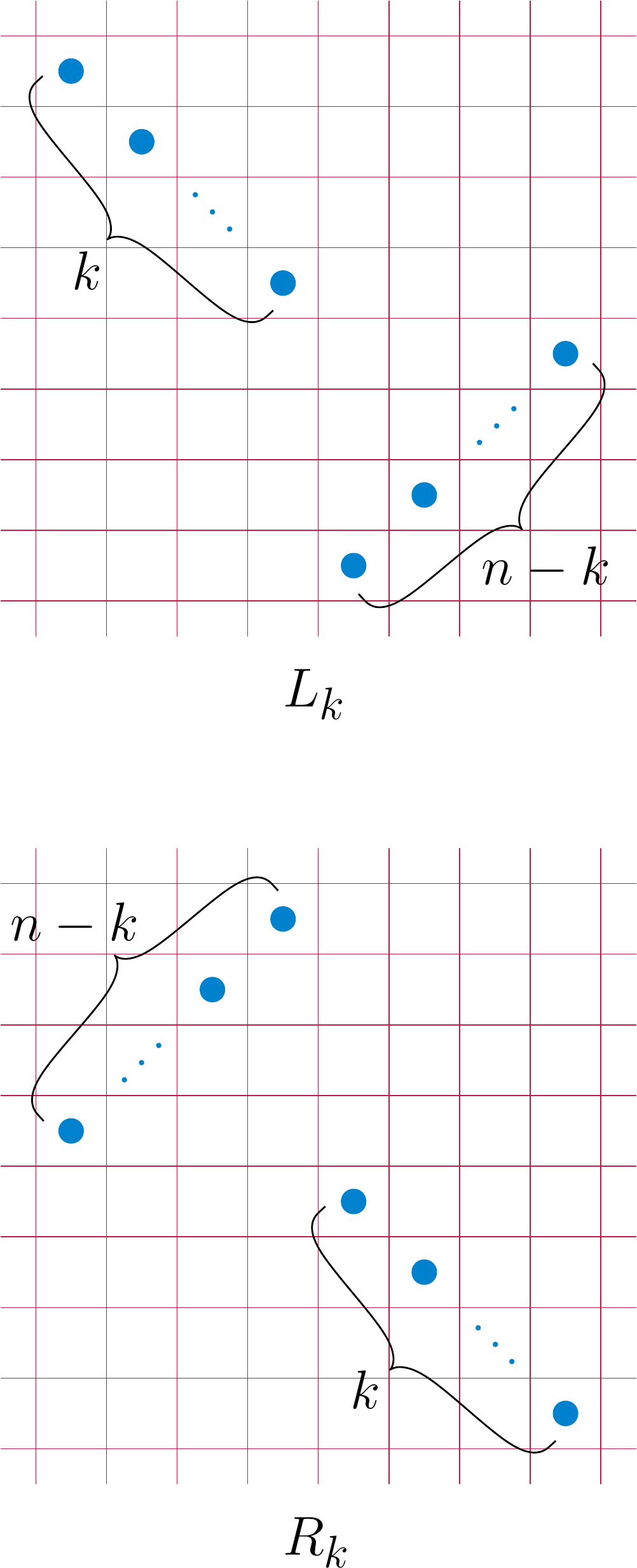}}
\hfill
\subfloat{\includegraphics[width=0.65\textwidth]{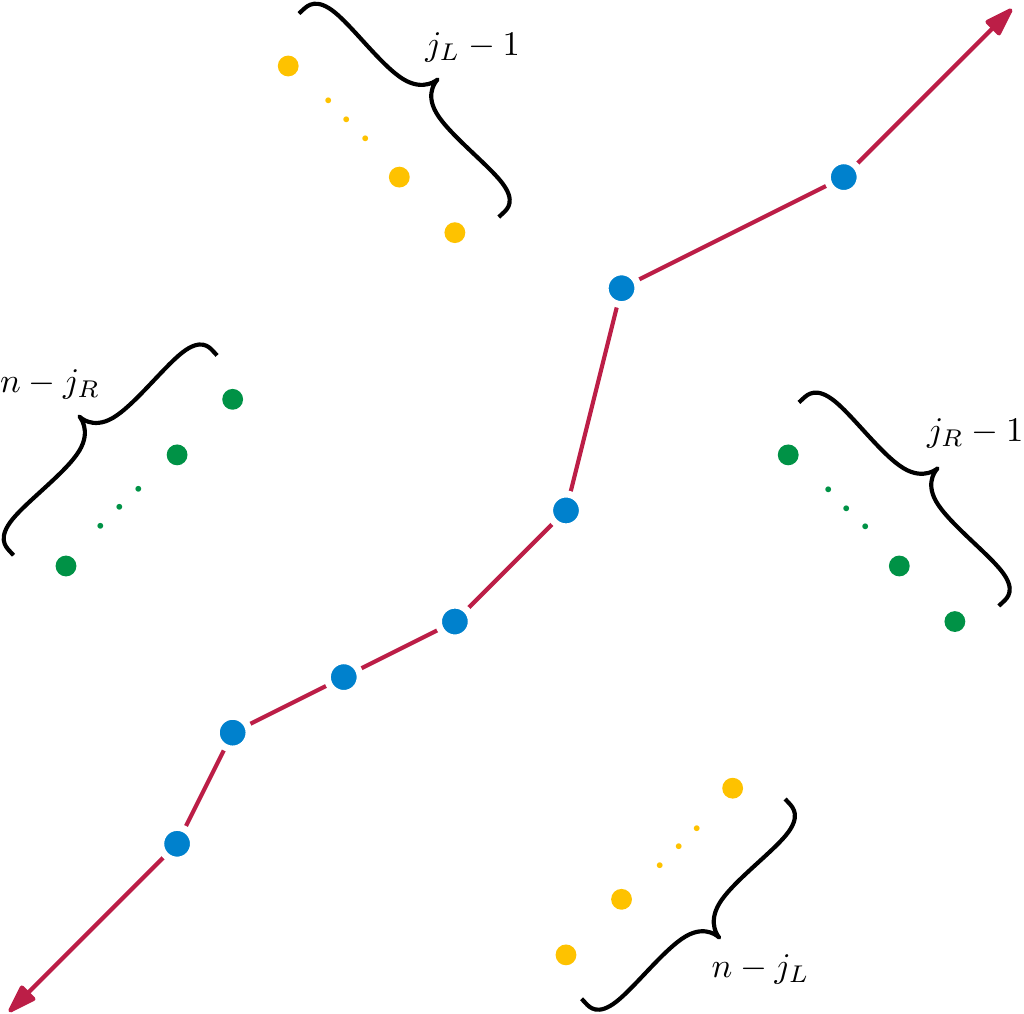}}

\caption{Left: The permutations we use to prove our lower bound for $S_n(213, 132, 3412, 4231)$-superpatterns. Right: The three sets of points we describe in the proof are shown respectively in blue, yellow, and green. The spine is shown in red.}
\label{fig:213_3412_optimality}
\end{figure}

\begin{theorem}
Every superpattern for $S_n(213,132,3412,4231)$, for $n\ge 2$, has length at least $3n-4$.
\end{theorem}

\begin{proof}
To show optimality, let $\pi$ be a $S_n(213,132,3412,4231)$-superpattern. Consider embedding the permutations $L_k = n\,(n-1) \dots (n-k+1)\, 1\, 2 \dots (n-k)$ and $R_k = (k+1)\, (k+2) \dots n\, k \dots 1$ for $k = 0$ to $n-2$. Note that $L_0 = R_0$ is just a length $n$ increasing sequence and $L_n = R_n$ a length $n$ decreasing sequence. Each of these permutations belongs to $S_n(213,132,3412,4231)$.

Consider the points where $L_0$ is embedded in $\pi$. Call the line segments connecting adjacent pairs of these points, with two rays of slope one tending towards the infinities in the top right and bottom left, the \emph{spine} (Figure~\ref{fig:213_3412_optimality}). We consider three (not necessarily disjoint) sets of points in $\pi$. The first set consists of the points of $\pi$ on the spine, and has size at least $n$.

For the second set, consider embedding $L_k$ for $k = 1$ to $n - 2$. There are two cases for each $L_k$: (1) the decreasing sequence of $L_k$ intersects the spine or has a point to the right of the spine, or (2) the decreasing sequence is completely to the left of the spine. Let $L = \{ k : L_k \text{ satisfies case (1)} \}$. If $L$ is empty, then the second set consists of the $n - 2$ points of the decreasing sequence of $L_{n-2}$, which must be to the left of the spine, plus an arbitrary point on the spine. If $L$ is not empty, let $i$ be the smallest index of an element of $L$. Here, the second set consists of the $i - 1$ points in the decreasing sequence of $L_{i-1}$, which must be to the left of the spine, plus the $n - i$ points in the increasing sequence of $L_i$, which must be to the right of the spine. In either case, the second set consists of $n - 1$ points, made up of $j-1$ decreasing points to the left of the spine and either one point on the spine (if $j=n-1$) or $n-j$ increasing points to the right of the spine (if $1 \le j \le n-2$).

The third set is defined in the same way as the second one, by finding the minimum index of a sequence $R_k$ whose embedding intersects or is to the left of the spine, and by selecting a descending set to the right of the spine and an ascending set on or to the left of the spine whose total length adds up to $n-1$.

The three sets have total size at least $n + 2(n-1) = 3n - 2$. A short case analysis based on the fact that an increasing sequence and a decreasing sequence of points can intersect in at most one point shows that they can have at most two points in common. Hence, the union of the three set have size at least $3n - 4$. This shows that $\pi$ must have at least $3n - 4$ points, as claimed.
\end{proof}

\begin{corollary}
For every permutation class $P$ with $S_n(213,132,3412,4231)\subset P\subset S_n(213,3412)$ and $n\ge 3$, the optimal length of a superpattern for $P$ is exactly $3n-4$.
\end{corollary}

\section{General subclasses of 213-avoiding permutations and bounded-pathwidth graphs}
\label{sec:general-subclasses}

The previous section described near-linear superpatterns for certain subclasses of the $213$-avoiding permutations.
In this section we generalize these results to all proper subclasses of the $213$-avoiding permutations. As we show, all such classes have superpatterns whose size is within a polylogarithmic factor of linear. Our superpattern construction is recursive, and combines ideas from Section~\ref{sec:213-132} (a majorizing sequence along the main diagonal of a chessboard representation) and Section~\ref{sec:213-3412} (a tree with a long path, each node of which has descendants on both sides of the path).
This leads to near-linear universal sets for families of planar graphs that include all planar graphs of bounded pathwidth.

\subsection{Tree augmentations}

Recall from Lemma~\ref{lem:213-characterization} that if $\pi$ is a $213$-avoiding permutation, then the chessboard graph of $\pi$ may be a forest rather than a single tree, and the chessboard representation of $\pi$ may have some squares containing numbers greater than one. However, both of these types of complication may be removed by adding additional elements to $\pi$.

\begin{definition}
Let $\pi$ be a $213$-avoiding permutation of length $n$. Then a \emph{tree augmentation} of $\pi$ is a $213$-avoiding permutation $\sigma$ that contains $\pi$, such that the chessboard graph of $\sigma$ is a tree and every square of its chessboard representation is either zero or one.
\end{definition}
An example of a tree augmentation is depicted in Figure~\ref{fig:tree-aug}. 

\begin{figure}[t]
\centering
\includegraphics[width=0.9\textwidth]{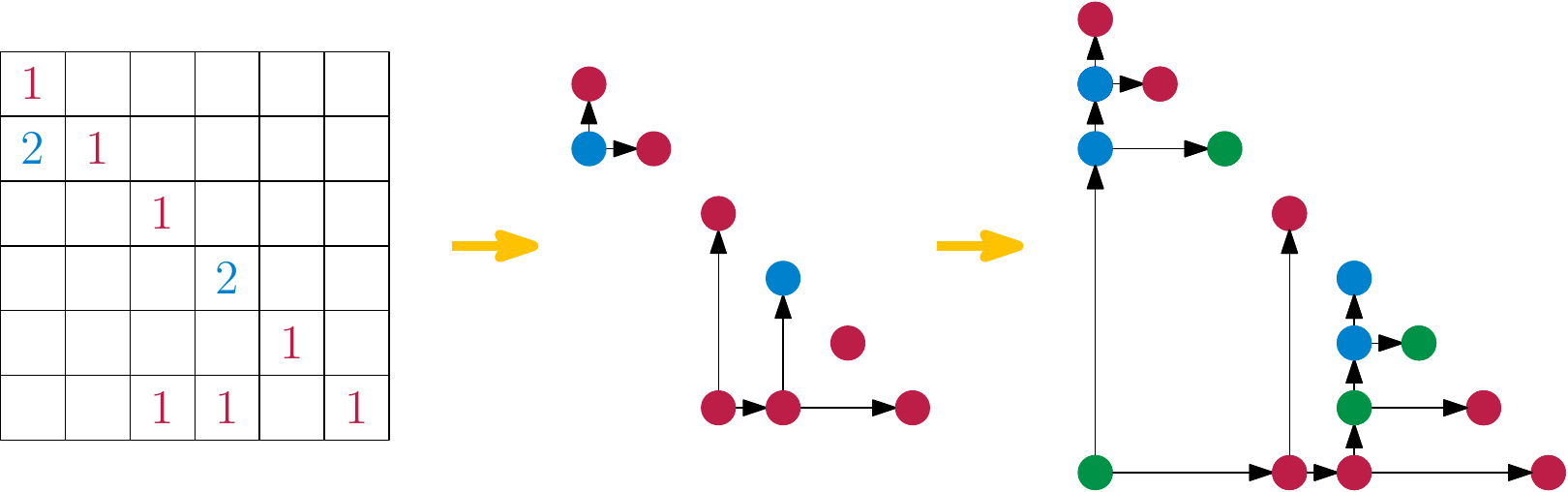}
\caption{The chessboard of a permutation, its corresponding chessboard graph, and its tree augmentation from Lemma~\ref{lem:tree-aug-size}.  Real vertices are colored red and blue, and fictitious vertices are colored green.}
\label{fig:tree-aug}
\end{figure}

\begin{definition}
Let $\sigma$ be a tree augmentation of a permutation $\pi$. Then nodes in the chessboard graph of $\sigma$ are called \emph{real} if they correspond to elements of $\pi$ and \emph{fictitious} if they were added in the augmentation process.
\end{definition}

\begin{lemma} \label{lem:tree-aug-size}
Every $213$-avoiding permutation $\pi$ has a tree augmentation of length at most $2n-1$.
\end{lemma}

\begin{proof}
For each chessboard square containing a number greater than one, adding an additional element along the main diagonal of the chessboard representation (creating an additional row and column) can split this nonzero into two smaller numbers, while preserving $\pi$ as a pattern. This adds $k-1$ new elements for each square with number $k$. Moreover, each tree in the chessboard graph may be connected to another tree by adding one additional element, again without affecting $\pi$ as a pattern. This shows that at most $n-1$ elements need to be added to create a tree augmentation of $\pi$.
\end{proof}

\subsection{Strahler number}
\begin{definition}
The \emph{Strahler number} of a node $x$ in a directed tree $T$ is a number defined by a bottom-up calculation in the tree, as follows: if $x$ is a leaf, its Strahler number is one. Otherwise, let $s$ be the largest Strahler number of a child of~$x$. If $x$ has only one child with Strahler number~$s$, its Strahler number is also~$s$, and if $x$ has multiple children with Strahler number~$s$, then its Strahler number is~$s+1$.
\end{definition}
Equivalently, the Strahler number of $x$ is the number of nodes on a root-to-leaf path of the largest complete binary tree that may be obtained from the subtree of $T$ rooted at $x$ by contracting edges (Figure~\ref{fig:strahler}).
\begin{figure}[t]
\centering
\includegraphics[width=0.9\textwidth]{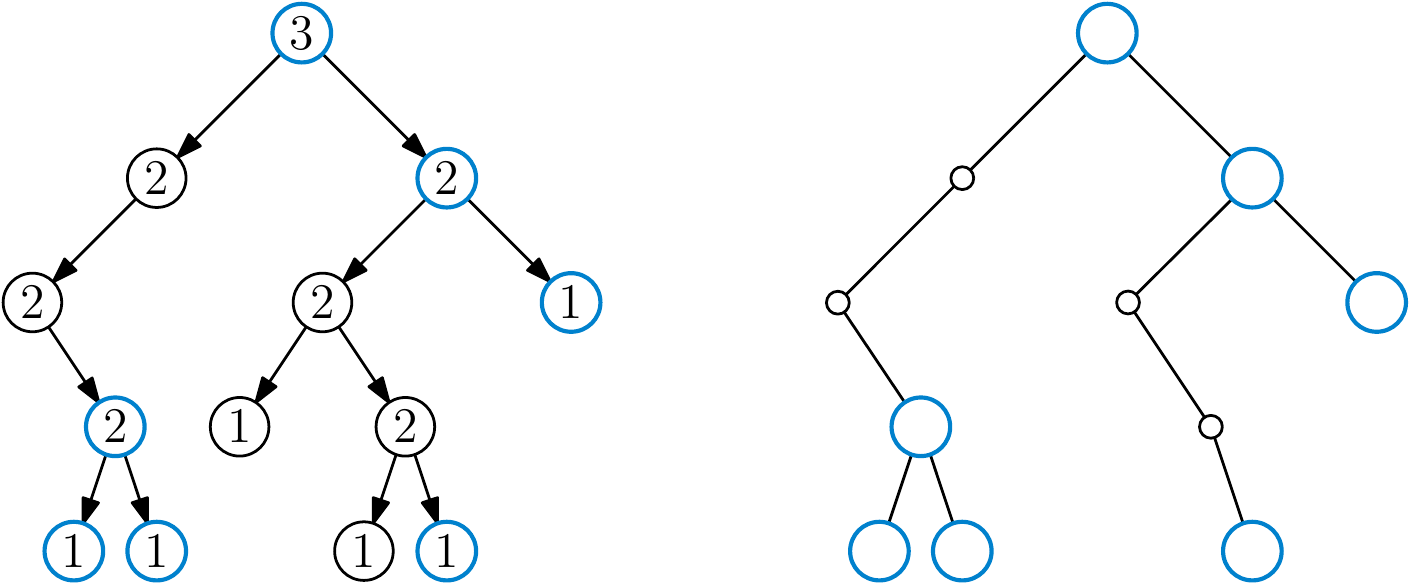}
\caption{A tree with nodes labeled by their Strahler number and one of the underlying complete binary trees.}
\label{fig:strahler}
\end{figure}

\begin{definition}
\label{def:Strahler-permutation}
If $\pi$ is a $213$-avoiding permutation, then we define the \emph{Strahler number of~$\pi$} to be the minimum Strahler number of any tree that can be obtained as the chessboard graph of a tree augmentation of~$\pi$.
\end{definition}
The only permutation with Strahler number one is the length-$1$ permutation. For every longer permutation, every tree augmentation must contain a nonleaf node; to prevent its subtree from collapsing into a single square of the chessboard representation, this node must have Strahler number at least two.

\begin{definition}
Let $\pi$ be a $213$-avoiding permutation with Strahler number $s$. A tree augmentation $\sigma$ of $\pi$ is called \emph{minimal} if the root of every subtree in the chessboard graph of $\sigma$ has the smallest Strahler number among all possible augmentations of elements of $\pi$ in that subtree, and if the length of $\sigma$ is the least among all such augmentations.
Note in particular, that the minimal tree augmentation of $\pi$ also has Strahler number $s$. 
\end{definition}

\begin{lemma}\label{lem:fake-child-real-parent}
Let $\sigma$ be a minimal tree augmentation of a permutation $\pi$. Then every fictitious leaf node must have a real parent in the chessboard graph of $\sigma$.
\end{lemma}

\begin{proof}
Suppose towards a contradiction that $v$ is a fictitious leaf node with a fictitious parent $u$ in the chessboard graph of $\sigma$. Observe that $u$ must have another child $w$; otherwise, the chessboard collapses to create a square with value greater than one. If $u$ has no parent, then $\pi$ must have length one and the contradiction is obvious. Otherwise, either (1) deleting $v$ and replacing $u$ by $w$, or (2) deleting $u$, moving $w$ to the column of the parent of $u$, and moving $v$ to the row of the parent of $u$ will produce a tree augmentation with both smaller length and smaller Strahler number at a root, contradicting the assumption that $\sigma$ is minimal.
\end{proof}

Given a node $v$ in a tree augmentation $T$ we say that $u$ is an \emph{immediate real descendant} of $v$ if $u$ is a descendant of $v$, a real node in $T$, and the path from $v$ to $u$ contains only fictitious nodes.

\begin{lemma}
\label{lem:Strahler-node-lemma}
Let $\pi$ be a $213$-avoiding permutation and $T$ be the tree obtained as the chessboard graph of a minimal tree augmentation of $\pi$. Then every node in $T$ of Strahler number $t \ge 3$ without a descendant of Strahler number $t$ has at least two immediate real descendants of Strahler number $t-1$.
\end{lemma}

\begin{proof}
Let $v$ be a node in $T$ with Strahler number $t \ge 3$ without any descendants of Strahler number $t$. Note that $v$ does not have any leaf children, so by Lemma~\ref{lem:fake-child-real-parent}, every path from $v$ to a leaf contains a real node. Let $R$ be the set of all immediate real descendants of $v$. Note that every node in $R$ has a distinct row and a distinct column in the chessboard; otherwise, a node would have two parents and $T$ could not be a tree. Let $u$ be a node in $R$ with the largest Strahler number. Consider the following chessboard graph $T'$ obtained by modifying the subtree of $T$ rooted at $v$: delete the fictitious nodes on every path from $v$ to a node of $R$, then add a fictitious node at (1) the intersection of the column of $v$ with every row containing a node of $R$ from the row of $v$ to the row of $u$, and (2) the intersection of the row of $u$ with every column containing a node of $R$ from the column of $v$ to the column of $u$ (see Figure~\ref{fig:Strahler-node-lemma}).

The graph $T'$ corresponds to a tree augmentation. For $v$ to have Strahler number $t$, at least two nodes in $R$ must have Strahler number $t-1$; otherwise, $v$ under $T'$ would have Strahler number $t-1$, contradicting that the original tree augmentation is minimal.
\end{proof}

\begin{figure}[t]
\centering
\includegraphics[width=0.3\textwidth]{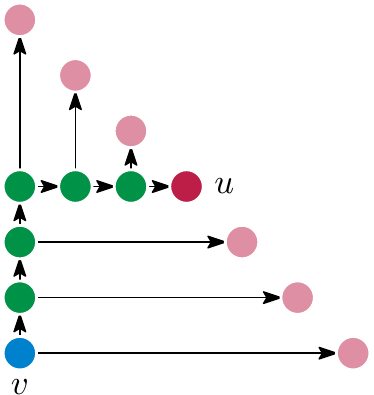}
\caption{Example of the augmentation in Lemma~\ref{lem:Strahler-node-lemma}. Immediate real descendants are colored red and the new fictitious nodes are colored green.}
\label{fig:Strahler-node-lemma}
\end{figure}

\begin{lemma}
\label{lem:Strahler-implies-tree}
If $\pi$ has Strahler number $s$, then $\pi$ contains a pattern whose chessboard graph is a complete binary tree of height $s-2$.
\end{lemma}
\begin{proof}
Let $T$ be a tree obtained as the chessboard graph of the minimal tree augmentation of $\pi$. By Lemma~\ref{lem:Strahler-node-lemma}, $T$ has a real node of Stahler number $s-1$ with two real descendants of Strahler number $s-2$. We may continue this process recursively to construct a complete binary tree of height $s-2$ as a minor of $T$ consisting of real nodes. The pattern of $\pi$ associated with these nodes satisfies the requirements of the lemma.
\end{proof}

\subsection{Superpatterns for small Strahler number}

If we parameterize the $213$-avoiding permutations by their Strahler number, then (as we show now) the permutations with bounded parameter values have superpatterns of near-linear size. As we will also show, every proper subclass of the $213$-avoiding permutations has bounded Strahler number. Therefore, the same near-linear bound on superpattern size applies to every proper subclass of the $213$-avoiding permutations.

\begin{figure}[t]
\centering
\includegraphics[width=0.9\textwidth]{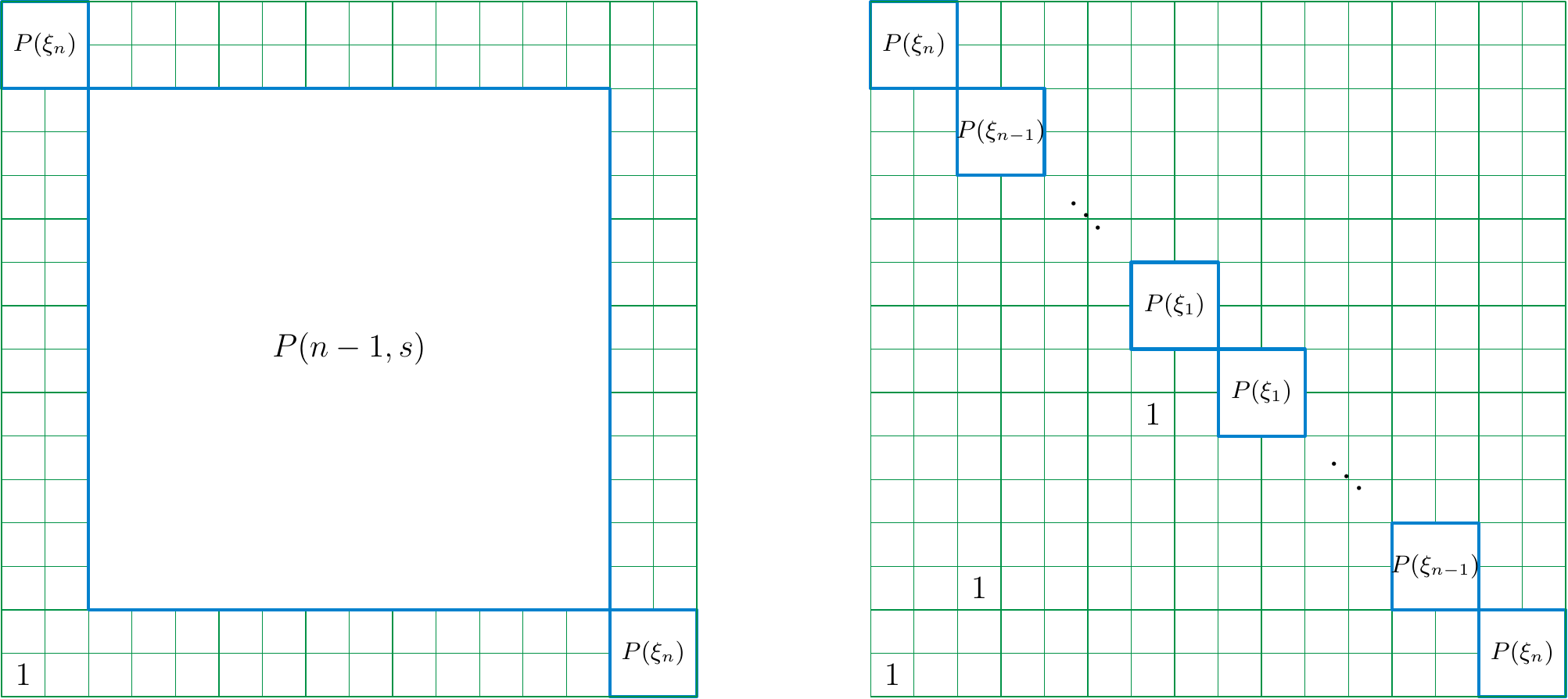}
\caption{Chessboard representation of the superpattern $P(n,s)$. Each $P(\xi_i)$ in the figure is $P(\xi_i, s-1)$, with the $s-1$ omitted because of space constraints.}
\label{fig:Strahler-construction}
\end{figure}

We construct our superpatterns by recursion. Let $P(n,s)$ denote the recursively constructed superpattern for minimally tree-augmented permutations of length $n$ and Strahler number at most $s \ge 2$. If a tree-augmented permutation has Strahler number~$2$ then its chessboard graph must be a caterpillar and the permutation must avoid the pattern $3412$. So for a base case we let $P(n,2)$ be the superpattern from Theorem~\ref{thm:213-3412-superpattern} of length $3n-4$. 

Let $P(0,s)$ be the empty permutation. Recursively define $P(n,s)$ for $n\ge 1$ and $s\ge 3$ to be the permutation formed by taking the length-$1$ permutation and adding above and to the right of it the three permutations $P(\xi_n, s-1)$, $P(n-1,s)$ and $P(\xi_n, s-1)$, arranged so that the left-to-right order of these permutations agrees with the top-to-bottom order. Figure~\ref{fig:Strahler-construction} (left) depicts one step of this construction, and Figure~\ref{fig:Strahler-construction} (right) depicts the whole pattern $P(n,s)$ in terms of patterns constructed in the same way for smaller values of~$s$.

\begin{theorem}
\label{thm:strahler}
For every fixed positive integer $s$, the $213$-avoiding permutations with Strahler number at most $s$ have superpatterns of length $O(n\log^{s-1}n)$.
\end{theorem}

\begin{proof}
Without loss of generality, we may consider solely the tree-augmented permutations, since a superpattern for the minimal tree augmentation for $\pi$ is also a superpattern for $\pi$ of the same asymptotic size. 

Let $\pi'$ be a tree-augmented permutation of length $n$. Then $\pi'$ may be found as a pattern in $P(n,s)$ as follows: let $p$ be a path from the root to a leaf of the chessboard graph of $\pi'$ that contains all nodes of Strahler number $s$, let $C_i$ denote the subgraph of the chessboard graph consisting of the $i$th node (counting starting from the leaf) on this path together with the child that is not on the path and all descendants of this child, and let $c_i$ denote the number of elements in $C_i$. Then $\sum c_i=n$, so we may apply Lemma~\ref{lem:sawtooth} to find a subsequence $t_i$ of the sequence $\xi$ with the property that, for all $i$, $c_i \le t_i$.
Let $T_i$ be the pattern $P(j,s)$ in our construction of $P(n,s)$ corresponding to $t_i$ (namely, $j$ is the index of $t_i$ in $\xi$). We may find $\pi'$ as a pattern in $P(n,s)$ by mapping the root node of $C_i$ to the length-$1$ permutation in $T_i$, and by mapping the elements in the remaining subtree of $C_i$ to one of the two copies of $P(t_i,s-1)$ (whichever copy is on the correct side of the root node). This shows that $P(n,s)$ is a superpattern for all minimal tree augmentations of length $n$.

The length bound on these superpatterns follows by the following straightforward calculation,
\begin{align*}
    |P(n,s)| &= 1 + 2|P(\xi_n,s-1)| + |P(n-1,s)|\\
             &= 1 + O(\xi_n \log^{s-2} \xi_n) + |P(n-1,s)|\\
             &= \sum_{k = 1}^n O(\xi_k \log^{s-2} \xi_k)\\
             &= O(n \log^{s-1} n),
\end{align*}
proving the theorem.
\end{proof}

\begin{corollary}
\label{cor:near-linear-for-all-proper-subclasses-of-213}
Let $\pi$ be an arbitrary $213$-avoiding permutation, and let $h$ be the number of nodes on the longest root-to-leaf path of a tree augmentation $\pi'$ of $\pi$, with the augmentation chosen to minimize~$h$.
Then the $\{213,\pi\}$-avoiding permutations have superpatterns of length $O(n\log^{h+1} n)$.
\end{corollary}

\begin{proof}
A $\{213,\pi\}$-avoiding permutation cannot have Strahler number $h+2$ or greater, for if it did then by Lemma~\ref{lem:Strahler-implies-tree} it would contain a pattern whose chessboard graph is a complete binary tree of height $h$. This pattern contains $\pi'$, and therefore also $\pi$, as a pattern. The result follows from Theorem~\ref{thm:strahler}.
\end{proof}

\subsection{Bounded pathwidth}

The \emph{pathwidth} of a graph $G$ is one less than the size of a maximum clique in an interval supergraph of $G$ chosen to minimize this maximum clique size~\cite{BodKlo-Algs-96,KirPap-DM-85,RobSey-JCTSB-83}.
Pathwidth is closely related to treewidth, which may be defined in the same way with chordal graphs in place of interval graphs. Both treewidth and pathwidth are monotonic under graph minor operations (vertex and edge deletion, and edge contraction). The pathwidth of a graph is therefore at least as large as the treewidth, and is also at most $O(\log n)$ times the treewidth~\cite{BodKlo-Algs-96}. Examples of graphs for which this $O(\log n)$ bound is tight include the complete binary trees, for which the treewidth is one and the pathwidth is $\Omega(\log n)$. As we show, the fact that these trees have high pathwidth allows us to apply our results on pattern-avoiding permutations to derive near-linear universal point sets for the planar graphs of bounded pathwidth.

\begin{lemma}
\label{lem:chessboard-to-starting-graph}
Let $G$ be a maximal planar graph together with a canonical representation, and let $T = \ctree(G)$, $\pi_G = \cperm(G)$ be respectively the canonical tree and permutation derived from that canonical representation, as in Section~\ref{sec:superpattern-to-universal}. Let $\tau$ be a pattern in $\pi_G$ whose chessboard graph is another tree $T'$. Then a tree isomorphic to $T'$ may be formed by contracting edges in $T$.
\end{lemma}

\begin{proof}
In $T'$, an element $y$ is a descendant of an element $x$ if and only if $y$ is greater than~$x$ both in sequence order and in value order; this means that, for the corresponding vertices in~$T$, $y$ is after~$x$ both in preorder and in reverse postorder. Recall from Section~\ref{sec:superpattern-to-universal} that this is only possible when $y$ is also a descendant of~$x$ in~$T$ as well as in~$T'$. To obtain $T'$ from~$T$, we need merely find each node~$z$ that does not belong to~$T'$, contract the edge from $z$ to its parent, and if necessary contract one more edge to the root of~$T$ (if that node does not belong to~$T'$).
\end{proof}

The following lemma allows us to deal with planar graphs of low pathwidth that are not themselves maximal planar.

\begin{lemma}
\label{lem:good-tree}
Let $G$ be a connected planar graph that is not maximal. Then there exists a maximal planar supergraph $G'$ of $G$, on the same vertex set, and a canonical representation of $G'$, such that if $T = \ctree(G')$ is derived from the canonical representation,
then all edges of $T$ that do not have $v_1$ as an endpoint belong to $G$.
\end{lemma}

\begin{proof}
We choose arbitrarily a base edge $v_1v_2$ in $G$.
Next, we construct the supergraph $G'$ and the canonical representation greedily, at each step maintaining a canonical representation of a subset $S$ of the vertices of $G$, and a supergraph of $G$ for which $S$ induces a triangulated disk $D$ having the base edge on its boundary ($D$ could be degenerate, consisting of only the subgraph induced by $v_1$ and $v_2$).
At each step of the construction, until $S$ contains all of the vertices of $G$, we choose one new vertex of $G\setminus S$ to add to $S$.

We order the vertices around the boundary of $D$ left-to-right from $v_1$ to $v_2$ along the path that does not include edge $v_1v_2$ (in the case $D$ is degenerate, we do include the edge $v_1 v_2$). Each vertex $u$ that does not belong to $S$ but is adjacent to $S$ has a set of neighbors along that path that lie within some interval from the leftmost neighbor to the rightmost neighbor (but possibly with vertices interior to the interval that are not neighbors of $u$). If $u_i$ and $u_j$ are two different vertices that do not belong to~$S$ but are adjacent to $S$, their intervals are either disjoint or nested. We distinguish two cases:
\begin{itemize}
\item If there exist vertices in $G \setminus S$ that are adjacent to a vertex in $S$ other than $v_2$, then choose $u$ to be one such vertex whose rightmost neighbor $v_i \ne v_2$ is as far to the right as possible, and whose leftmost neighbor is also as far to the right as possible (breaking ties arbitrarily).
If $u$ has at least two neighbors in $S$, we add to $G'$ edges between $u$ and all the vertices on $D$ between its leftmost and rightmost neighbors; this cannot cause $G'$ to become nonplanar because the nesting property of the intervals ensures that these vertices on $D$ have no other neighbors outside of $D$. If $u$ has only $v_i$ as its neighbor, we add to $G'$ an edge from $u$ to the next vertex to the right of $v_i$ on the boundary of $D$. Then, we add $u$ to the canonical representation.
\item If all vertices in $G \setminus S$ adjacent to $S$ have only $v_2$ as a neighbor, then choose $u$ arbitrarily among such vertices. We add to $G'$ an edge from $u$ to every other vertex on the boundary of $D$. Because only $v_2$ has neighbors outside of $S$, this again cannot cause $G'$ to become nonplanar. Then, we add $u$ to the canonical representation.
\end{itemize}
The result after either of these two cases is a canonical representation of a triangulated disk with one more vertex than before. After repeating $n-2$ times, all vertices must belong to the triangulated disk. We complete the remaining graph to a maximal planar graph by adding edges from the final vertex $v_n$ to all vertices that share a face with it.

Under the constructed canonical representation, $\ctree(G')$ consists of the leftmost incoming edge for every vertex other than $v_1$. This edge either belongs to~$G$ (if the vertex was added by the first case or is $v_2$) or has $v_1$ as an endpoint (if the vertex was added by the second case or is $v_n$).
\end{proof}

\begin{theorem}
For every constant $w$, the planar graphs of pathwidth $w$ have universal point sets of size
$O(n\log^{O(1)} n)$.
\end{theorem}

\begin{proof}
Let $T$ be a complete binary tree of sufficiently large size that the pathwidth of $T$ is greater than~$w$, and $T'$ be a tree whose root has a single child subtree isomorphic to $T$. Suppose $G$ is a planar graph of pathwidth $w$. We apply Lemma~\ref{lem:good-tree} to augment $G$ to a maximal planar graph $G'$ with a canonical representation having the property given in the lemma. Because pathwidth is closed under minors, if $G$ is a graph of pathwidth at most $w$, then it cannot contain $T$ as a minor. Moreover, $\ctree(G')$ cannot have a minor isomorphic to $T'$, for if it did, then by removing the root edge (which might not belong to $G$), $G$ would have a minor isomorphic to $T$, which is not possible.

Let $\tau$ be the pattern whose chessboard graph is $T'$. Since $T'$ is not a minor of $G'$, it follows from Lemma~\ref{lem:chessboard-to-starting-graph} that $\tau$ is not a pattern in $\cperm(G')$. By Corollary~\ref{cor:near-linear-for-all-proper-subclasses-of-213}, the permutations that avoid both $213$ and $\tau$ have a superpattern $\sigma$ of size $O(n \log^{O(1)} n)$. Applying the same ideas as in the proof of Theorem~\ref{thm:U_sigma_universal}, $\stretchperm(\augment(\sigma))$ is a universal point set of the same asymptotic size for planar graphs of pathwidth $w$.
\end{proof}

\section{Simply-nested planar graphs}

Angelini et al.~\cite{AngDibKau-GD-2012} define a \emph{simply-nested planar graph} to be a graph in which a tree is surrounded by a sequence of chordless cycles, which form the \emph{levels} of the graph (the tree is level zero); additional edges are allowed connecting one level to another. As they showed, if the number of vertices~$n_i$ in the $i$th level is fixed, for all~$i$, then one can define a universal point set with $8\sum n_i$ vertices, consisting of $n$ concentric circles with $8n_i$ equally spaced points on the $i$th circle and with carefully chosen radii.
They use this result to prove an $O\left(n (\log n / \log\log n)^2\right)$ bound on the size of universal point sets for simply-nested planar graphs, which we improve to $O(n\log n)$ using our results on subsequence majorization.

\begin{theorem}
\label{thm:simply-nested}
There is a universal point set of size $O(n\log n)$ for the simply-nested planar graphs.
\end{theorem}

\begin{proof}
We use the sequence of nested circles given by Angelini et al. for the version of the problem with fixed values of $n_i$, and for the assignment $n_i=\xi_i$, $i=1,2,\dots n$.
The resulting point set has $O(n\log n)$ points by Lemma~\ref{lem:brockhaus}.
For an arbitrary $n$-vertex simply-nested planar graph, the numbers of vertices per level can be majorized by a subsequence of $n_i$, and the vertices of each level can be assigned to the corresponding circle of the point set. Compared to the vertex placement of Angelini et al., this vertex placement skips some of the nested circles (the ones not used in the subsequence of $\xi_i$) and has more than the necessary number of points on some circles, neither of which causes any difficulty in the placement.
\end{proof}

Among the planar graphs covered by this result are the \emph{squaregraphs}~\cite{BanCheEpp-SJDM-10}, planar graphs that have an embedding in which each bounded face is a quadrilateral and each vertex is either on the outer face or has four or more neighbors. For instance, the subset of the integer lattice on or inside any simple lattice cycle (a \emph{polyomino}) forms a graph of this type. Although they are not themselves simply-nested, by Lemma 12.2 of Bandelt et al.~\cite{BanCheEpp-SJDM-10}, the squaregraphs can be embedded as subgraphs of simply-nested planar graphs on the same vertex set (possibly allowing some levels of the nesting structure to be degenerate cycles with one or two vertices). Therefore, by Theorem~\ref{thm:simply-nested}, the squaregraphs have universal point sets of size $O(n\log n)$.

\section{Conclusion}
In this paper we have constructed universal point sets of size $n^2/4 - \Theta(n)$ for planar graphs, and of subquadratic size for graphs of bounded pathwidth and simply-nested planar graphs. In the process of building these constructions we have provided a new connection between universal point sets and permutation superpatterns. We have also, for the the first time, provided nontrivial upper bounds and lower bounds on the size of superpatterns for restricted classes of permutations. We leave the following problems open for future research:

\begin{itemize}
\item Which natural subclasses of planar graphs  (beyond the bounded-pathwidth graphs) can be represented by permutations in a proper subclass of $S_n(213)$?
\item Can we reduce the gap between our $O(n^2)$ upper bound and $\Omega(n\log n)$ lower bound for $S_n(213)$-superpatterns?
\item Our construction uses area exponential in $n^2$; how does constraining the area to a smaller bound affect the number of points in a universal point set?
\end{itemize}

\clearpage
{\raggedright
\bibliographystyle{abbrvurl}
\bibliography{universal}}

\end{document}